\documentclass{article}
\usepackage{arxiv}
\usepackage[utf8]{inputenc} %
\usepackage[T1]{fontenc}    %
\usepackage{hyperref}       %
\usepackage{url}            %
\usepackage{booktabs}       %
\usepackage{amsfonts}       %
\usepackage{nicefrac}       %
\usepackage{microtype}      %
\usepackage{lipsum}

\usepackage{graphicx}
\usepackage[dvipsnames]{xcolor}
\usepackage{hyperref}
\usepackage{float}
\usepackage{multirow}
\usepackage{longtable}
\usepackage{amsmath}
\usepackage{amsthm}
\newtheorem{theorem}{Theorem}
\newtheorem{lemma}{Lemma}
\newtheorem*{obs*}{Observation}
\usepackage{times}
\usepackage{subcaption}
\usepackage{mathrsfs}
\usepackage[font=small,labelfont=bf]{caption}
\usepackage{rotating}
\usepackage{siunitx}
\usepackage{lineno}

\graphicspath{{./PaperFigures/}}

\newcommand{\fcaptionsize}{\small}

\title{Interventions with Inversity in Unknown Networks Can Help Regulate Contagion} %

\author{
  Vineet Kumar \\
  Yale School of Management\\
  Yale University\\
  \texttt{vineet.kumar@yale.edu} \\
   \And
 David Krackhardt \\
  Heinz School of Public Policy and Management\\
  Carnegie Mellon University\\
  \texttt{krack@cmu.edu} \\
   \And
 Scott Feld \\
  Department of Sociology\\
  College of Liberal Arts\\
  Purdue University\\
  \texttt{sfeld@purdue.edu} \\
}

\begin{document}
\maketitle

\begin{abstract}
Network intervention problems often benefit from selecting a highly-connected node to perform interventions using these nodes, e.g. immunization. However, in many network contexts, the structure of network connections is unknown, leading to a challenge. We develop and examine the mathematical properties of two distinct informationally light strategies, a novel global strategy and local strategy, that yield higher degree nodes in virtually any network structure. We further identify a novel network property called Inversity, whose sign determines which of the two strategies, local or global, will be most effective for a network. We demonstrate that local and global strategies obtain a several-fold improvement in node degree relative to a random selection benchmark for generated and real networks (including contact, affiliation and online networks). In some networks, they achieve a 100-fold improvement. We show how these new strategies can be used to control contagion of an epidemic spreading across a set of village networks, finding that the strategies developed here require far fewer ($<50\%$) nodes to be immunized, relative to the random strategy baseline. Prior research has typically used the complete network structure to choose nodes for optimal seeding. The relevant network is often costly to collect, and is privacy-invasive, requiring knowing each person's network neighbors, and might not be possible to obtain for time-sensitive interventions. Our interventions are less invasive of individual privacy, since each selected node only needs to nominate some network neighbors for intervention, while mathematically guaranteed to provide better connected nodes.
\end{abstract}


\section{Introduction}
Network-based interventions are of crucial importance in any setting where an individual’s choice or action has an indirect impact on others. There are a wide range of network intervention applications. Consider: (a) A new infectious disease is spreading through a large population. We want to minimize the number of infected individuals by inoculating using a new vaccine; however, we only have a limited number of doses to administer. (b) We have a limited number of free samples of new product to distribute to consumers, so they can share information through word of mouth, and we would like to maximize the number of consumers who receive word of mouth. (c) We would like to identify virally spreading contagion (informational or biological) as quickly as possible by choosing individuals as observation stations (or for contact tracing).

Although seemingly distinct, these problems (a)-(c) represent a class of network interventions  \cite{valente2012network} in which we benefit from identifying more central or highly connected individuals in the network.\footnote{They belong to the class of ``simple contagion'' problems, which require only one rather than ``complex contagion'' that require multiple exposures \cite{centola2007complex,centola2010spread}.}
However, the challenge is that we do not have access to the \textit{relevant network structure}. In application (a), having the Facebook network structure might not be useful, since the relevant network would be the \textit{physical contact} network. In contrast, for  application (b), finding a high degree node using a physical contact network of everyone who interacts with a physician is unlikely to be fully informative in characterizing her opinion leadership in the profession.  For (c), and in epidemics like those analyzed in \S \ref{section:application}, carrying out contact tracing for all individuals can be expensive in effort and time.

We develop, model and characterize \textit{novel network intervention strategies}, which obtain highly connected nodes from a network by querying randomly chosen individuals, and which do not require access to the complete underlying network structure. These interventions termed  global and local strategies are based on the friendship paradox, and result from theoretical network properties developed here. We also derive a structural network property called inversity, which determines the relative effectiveness of the local and global strategies. These strategies have several advantages for implementation. First, they are informationally-light, in the sense that we don’t need to know anything about the network structure, and yet they provide provable advantages for almost all network structures, in contrast to most prior work that often requires complete knowledge of the network structure. The network may not be possible to obtain in a timely manner, or may vary over time. In contrast, for the strategies here, we need to identify initial individuals and obtain a (small) proportion of their friends to serve as seeds. Thus, the strategies can be implemented quickly since they only require local network information obtained by querying individuals or interaction data. Second, unlike other methods to obtain highly-connected nodes, our intervention strategies are much more privacy-sensitive and do not require the time and effort to map out the entire network. Third, the class of interventions here can be used in advance, i.e. for both prevention and treatment interventions.

\section{Friendship Paradox}
The Friendship Paradox, which our interventions are based on, is colloquially stated as ``your friends have more friends than you''  \cite{Feld1991,Zuckerman2001}.\footnote{The phenomenon has also been generalized to the idea that individual attributes and degree are correlated \cite{jo2014generalized}, e.g. an individual's co-authors are more likely to be cited \cite{eom2014generalized}, or that friends more active on social media\cite{hodas2013friendship}.} The intuition for why the friendship paradox helps obtain well-connected nodes is this: there are few well-connected hubs in real networks, and since they are connected to many other nodes (by definition), obtaining a friend (or neighbor) of a random node is likely to result in a hub with greater likelihood, compared to the case of randomly selecting nodes.
We establish that the friendship paradox is actually not just one statement, but a set of distinct claims (All theorems and proofs are in Supplement \S \ref{section:mathproofs}). First, we find an impossibility, i.e. the individual-level friendship paradox cannot hold for all individuals in a network (Theorem \ref{theorem:individualfp}). In practice for real networks, it can hold for a large proportion of nodes in the network (Figs. \ref{figure:fpindividual} and \ref{figure:individualfpcdf} in Supplement \S \ref{section:sindividualfp}). Second, we demonstrate that in contrast to this individual view, the \textit{average} number of friends of friends across the network can be characterized in two ways using the local and global mean defined below. Third, we find that both local and global means are greater than the mean degree of the network, and these means are related through a novel network characteristic we call inversity.

\paragraph*{Local and Global Means}
We formally characterize the two distinct but related network properties deriving from the friendship paradox relating to the ``average number of friends of friends." Denote a network (see Table \ref{table:notation-1} for full notation) as an undirected graph $\mathcal{G}=\left(V,E\right)$
with $V$ the set of nodes and $E$ the set of edges ($e_{ij}\in\left\{ 0,1\right\} $
denoting absence or presence of a connection between $i$ and $j$),
  $D_{i}$ refers to the degree of node $i$, and $N(i)$ the set of $i$'s neighbors.
We specify the local mean as:
\begin{equation}
\mu_{L}=\frac{1}{N}\sum_{i\in V} \left[ \dfrac{1}{D_i} \sum_{j\in N(i)}D_j \right]
\end{equation}
The global mean is defined as the
ratio of the total number of friends of friends to the total number
of friends in the network, consistent with \cite{Feld1991}:%
\begin{equation}
\mu_{G}=%
\dfrac{\sum_{i \in V}{\left[ \sum_{j\in N(i)}D_j \right]}}{\sum_{i \in V}D_i} 
\end{equation}
The above means arise from differently weighting the average degree across friends.
Both means above are consistent with the notion of ``average number of friends
of friends,'' although they are distinct network properties (see Fig. \ref{figure:localglobalexplanation} for an example and detailed explanation).
The global mean was theoretically investigated earlier and found to be greater than the average degree and is independent of the local structure of connections, given node degrees (Theorem \ref{theorem:globalmean}). In contrast, local mean has distinct properties that depend on local network structure (i.e. who is connected to whom). However, it does share the property with the global mean that it is greater than the mean degree (Theorem \ref{theorem:localmean}).

\renewcommand{\baselinestretch}{1}
\begin{sidewaysfigure}
\includegraphics[scale=0.7]{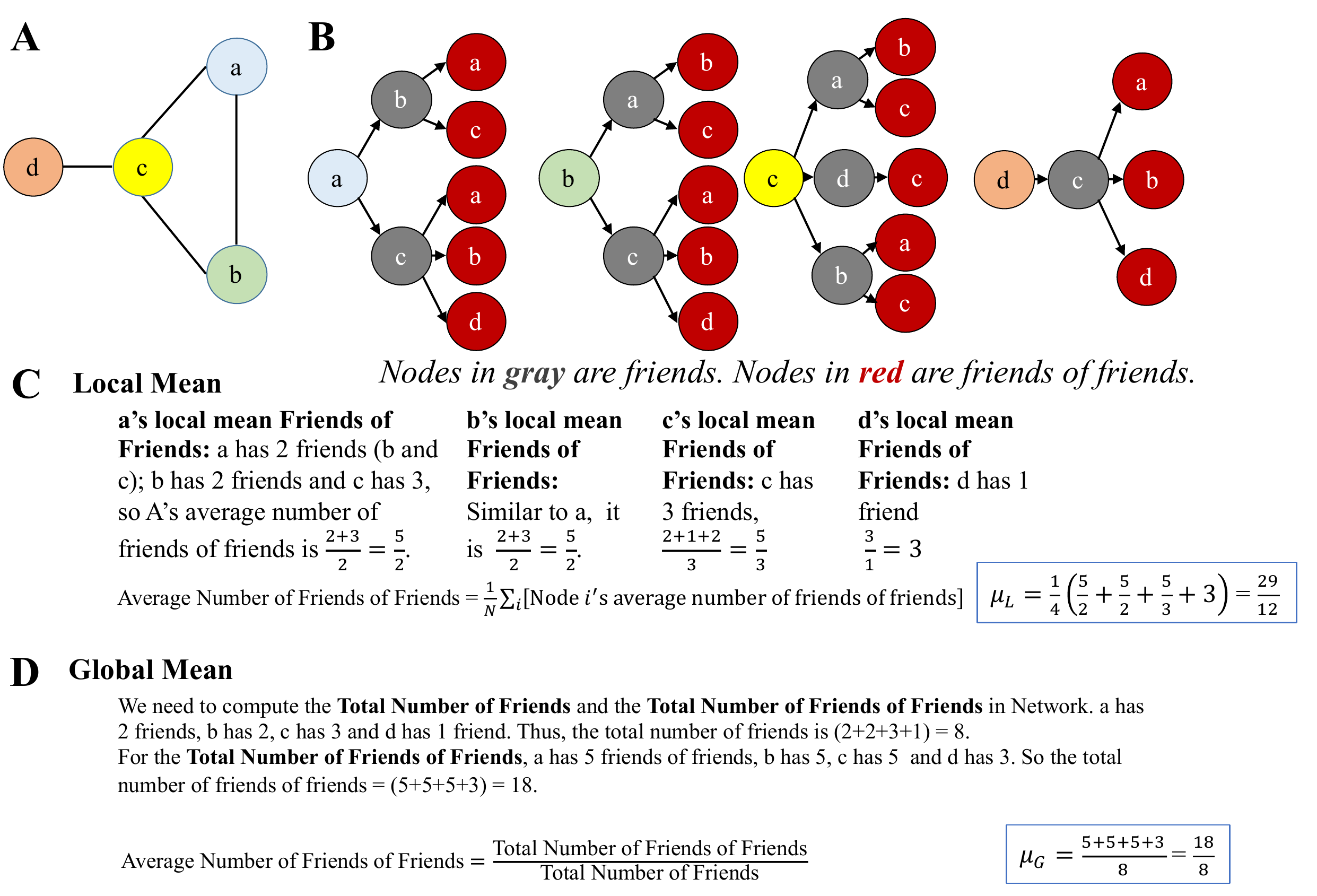}
\centering{}
\caption{\fcaptionsize{Local and Global Means in Example Network.
\textbf{(A) Network.} Example network with 4 nodes $a$,$b$,$c$ and $d$. 
\textbf{(B) Illustration of Friends and Friends of Friends.}  Each node is mapped out with its friends and friends of friends. The node is in light blue color, Friends are in gray color and Friends of Friends are in red color.  Node $a$ has 2 friends $b$ and $c$. Node $a$ also has 5 friends of friends. 
\textbf{(C) Local Mean.}  $a$ has 2 friends, $b$ and $c$.  
The total number of friends of friends in $a$'s network is 5: ($a$'s friend $b$ has 2 friends, $a$ and $c$; $a$'s friend $c$ has 3 friends, $a$, $b$ and $d$).  So average number of friends of friends for $a$ is the ratio of the number of nodes in black to the number of nodes in red, i.e. $\frac{2+3}{2}=\frac{5}{2}$. Similarly, for the other nodes, we have $b: \frac{2+3}{2}=\frac{5}{2}$, $c: \frac{2+1+2}{3}=\frac{5}{3}$, $d: \frac{3}{1}=3$. The local mean of the network is the mean of these local average friends of friends, so $\mu_L=\frac{1}{4}\left(\frac{5}{2}+\frac{5}{3}+\frac{1}{4}+\frac{3}{1}\right)=2.42$. \textbf{(D) Global Mean.} Global mean is the ratio of the total number of friends of friends to the total number of friends. The total number of friends of friends contributed by $a$ is 5. Similarly, $b$ contributes 5, $c$ contributes 5 and $d$ contributes 3 friends of friends.
Thus, the total number of friends of friends (i.e. the nodes in red color) are $(5+5+5+3)=18$.
The total number of friends is represented by the nodes in gray, $(2+2+3+1)=8$.
}}
\label{figure:localglobalexplanation}
\end{sidewaysfigure}
\renewcommand{\baselinestretch}{2}

The global mean is greater when there is higher \textit{variation across nodes} in terms of degree (variance), whereas local mean is greater when we have higher \textit{variation across edges}, i.e. when edges connect nodes of very dissimilar degree (e.g. with a hub and spoke network). More specifically, the global mean  is invariant to rewiring the network while keeping the degree distribution the same, whereas the local mean is impacted by the rewiring (Theorem \ref{theorem:rewire}). We term these means local or global since the former depends on the local structure (who is connected to whom), whereas the latter only depends on the global network properties (degree distribution).

We identify network structures that result in a greater divergence between these means and the average degree (Figure \ref{figure:two-by-two} in \S \ref{section:slocalglobalmeans}). We also examine a number of questions about the relationship between the local and global mean, including whether one of the means is always greater than the other, whether they exhibit correlated variation away from the mean degree (\S \ref{section:slocalglobalmeans}).

\section{Intervention Strategies\label{section:Intervention-Strategies}}
The above formulation of local and global mean suggests distinct intervention strategies. We illustrate \textit{random}, \textit{local} and \textit{global} strategies  to choose a ``seed'' node in the network beginning with an initial randomly chosen node (Table  \ref{table:strategyimplementation}). Observe that with the local strategy, the number of seed nodes is fixed, whereas it is probabilistic under the global strategy. For the local strategy, by construction, the expected degree of the obtained seed is equal to the local mean. For the global strategy, we prove that the expected degree of chosen nodes is equal to the global mean (Theorem \ref{theorem:globalintervention}).

Though the local and global strategies appear to be similar in the
sense that we are choosing friends of randomly chosen individuals,
the crucial distinction lies in whether we are choosing \emph{one}
random friend or whether we are choosing among \emph{each}
 friend. Table \ref{table:strategyimplementation} details the algorithms to obtain $k$ seeds in a network of size $N \gg k$. This impacts their relative effectiveness as examined in \S \ref{section:leverage} and \S\ref{section:application}.

\renewcommand{\baselinestretch}{1} 
\begin{table}[H]%
\caption{Implementation of Seeding Strategies}\label{table:strategyimplementation}
\parbox{0.60\textwidth}{
\begin{footnotesize}
\begin{tabular}{p{0.02\textwidth}p{0.02\textwidth}p{0.56\textwidth}}
\toprule
\textbf{Step}  & & \textbf{Details} \tabularnewline
\midrule
\textbf{0} &  & Fix $p \in (0,1]$ (only used for Global strategy in Step \textbf{2G}).
\tabularnewline
\midrule
\multicolumn{3}{l}{Repeat Steps 1-2 below until at least $k$ seeds are present in the seed set $\mathcal{S}$.}
\tabularnewline
\tabularnewline
\textbf{1} &  & Draw a random node $r$ uniformly from set of nodes, $V$.
\tabularnewline
 &  &In Example Network, Nodes $10$, $18$ and $12$ (in black) are drawn for (R), (L) and (G) strategies respectively.
\tabularnewline
\textbf{2} && Depending on the strategy Random (R), Local (L) or Global (G), do the following:
\tabularnewline
&\textbf{2R} & (\textit{Random}): Add $r$ to the seed set $\mathcal{S}$.
\tabularnewline
&&
\textcolor{darkgray}{In Example Network, add node $10$ to the seed set.}
\tabularnewline
&\textbf{2L} & (\textit{Local}): Obtain a node $s$ chosen with uniform probability from $r$'s friends, i.e. $s \in \mathcal{N}_r$.
\textcolor{darkgray}{Add the friend $s$ to the seed set $\mathcal{S}$.}
\tabularnewline
&&
\textcolor{red}{In Example Network, one of node $18$'s friends, node $14$ (in red), is chosen at random. Add node $14$ to the seed set.}
\tabularnewline
&\textbf{2G} & (\textit{Global}): For each of $r$'s friends, $s \in \mathcal{N}_r$: 
With probability $p$ ($0 < p \le 1$), add $s$ to the seed set $\mathcal{S}$.
\tabularnewline
&&
\textcolor{OliveGreen}{In Example Network, each of node $12$'s friends, nodes $1$, $2$, $8$ and $9$ (in green), are added probabilistically (with probability $p$) to the seed set.}
\tabularnewline

&&\textit{Implementation:} For each $s \in \mathcal{N}_r$, draw from an independent uniformly distributed random variable $z_s \sim U[0,1]$. If $z_s < p$, add $s$ to the seed set $\mathcal{S}$.
\tabularnewline
\midrule
\multicolumn{3}{p{\linewidth}}{Note: With Random and Local strategies, we will obtain exactly $k$ nodes in the seed set $\mathcal{S}$. With the global strategy we might obtain more than $k$ nodes in the seed set. In such a case, we select $k$ nodes at random from the seed set $\mathcal{S}$ without replacement.}
\tabularnewline
\bottomrule
\end{tabular}
\end{footnotesize}
}
\hspace{0.05in}
\qquad
\begin{minipage}[c]{0.35\textwidth}%
\centering
\vspace{-0.25in}
    \includegraphics[width=1.1\textwidth]{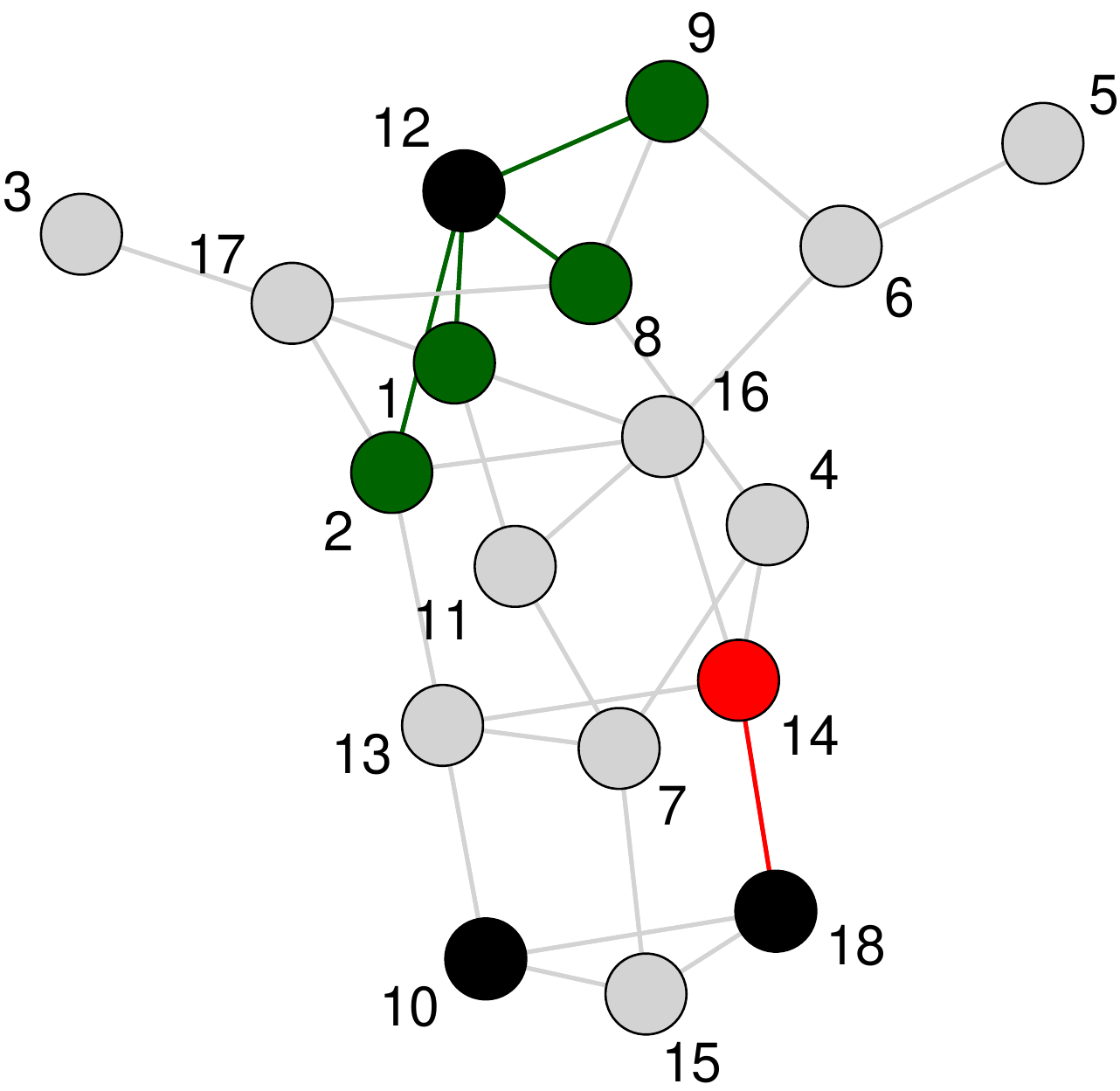}
\caption*{Example Network}
\label{fig:figure}
\end{minipage}
\tabularnewline

\end{table}

\renewcommand{\baselinestretch}{2}

\section{Inversity: Connecting Local and Global Means}\label{section:Inversity}
Since both local and global strategies can be used for interventions, we next characterize their relative effectiveness. We identify and define a novel network property, Inversity, that determines when the local mean is greater than the global mean. This property captures all local network information related to the local mean and is scale-invariant, i.e. independent of the size or density of the network. We find that the sign of inversity helps us determine which of the local mean or the global mean is higher for any given network. We show how inversity is related to but distinct from degree assortativity (in Supplement \S \ref{figure:inversityassortativity}).

Inversity is a correlation-based metric that
relates the global and local means for any network is obtained as follows.
First, define the following edge-based distributions to examine
the relationship between the means. The \emph{origin} degree (\textbf{O}),
$D^{O}(e)$, \emph{destination} degree (\textbf{D}), $D^{D}(e)$,
and \emph{inverse destination} degree (\textbf{ID}) distribution,
$D^{ID}(e)$, are defined across directed edges $e\in\hat{E}$
as:
$ D^{O}(e_{jk})=D_{j}, D^{D}(e_{jk})=D_{k}, D^{ID}(e_{jk})=\frac{1}{D_{k}}
$. 
We define the \emph{inversity} across the edge distribution
as the Pearson correlation across the origin and inverse degree distributions.
\begin{equation}
\rho=Corr\left(D^{\mathbf{O}},D^{\mathbf{ID}}\right)%
\end{equation}
We can then connect (see Theorem \ref{theorem:inversity}) the local and global means with \textit{inversity} and the degree distribution $\left(\kappa_m=\sum_{i \in V} D_i^m\right)$ as:
\begin{equation}
\mu_{L}   =  \mu_{G}+\rho \  \Psi (\kappa_{-1},\kappa_{1},\kappa_{2},\kappa_{3})
\end{equation}
where $\Psi$ is a positive function of the degree distribution.

Therefore, if inversity is known, we don't need the entire degree distribution to obtain the local mean. Rather, \emph{four} moments
of the degree distribution are sufficient for that purpose. %
Inversity captures the local information on imbalances in degree of nodes across edges, whereas the moments of the degree distribution represent global information about the network. Inversity $\rho$ has a critical role in determining whether the local or global mean is larger for a network; specifically,
 $\rho<0$ indicates the global mean is higher than the local mean, whereas $\rho>0$ indicates the reverse,  implying that knowing inversity can help us determine which strategy to use. Even computing inversity is information-light, requiring only the $2k$ distribution, which represents the degrees of nodes at the termini of each edge, rather than the entire network \cite{orsini2015quantifying}.

\section{Effectiveness of Strategies: Leverage}\label{section:leverage}
To identify how much of an improvement over the random strategy is possible, and how this varies across a variety of generated and real networks, we examine the relative effectiveness of strategies, with the random strategy as the baseline and characterize leverage as the improvement another strategy can obtain in terms of expected degree.
Leverage for strategy $s$ on network $\mathcal{G}$ is defined as $\lambda_s(\mathcal{G}) = \frac{\mu_s(\mathcal{G})}{\mu_D(\mathcal{G})}$ for $s\in\{L,G\}$ (since the random strategy obtains the mean degree in expectation, the leverage for $R$ is $1$ and it serves as a baseline). A star (or hub-spoke) network obtains the highest possible leverage (see Theorem \ref{theorem:leveragemax}).

\paragraph{Generated Networks:} \label{sub:strategies-generated} 
Networks generated from a number of commonly used generative mechanisms are used to assess a number of structural features with regard to the friendship paradox. We examine 3 generative mechanisms for networks: (a) Erdos-Renyi (ER) \cite{erdos1959random}, (b) Scale Free (SF) \cite{barabasi1999emergence} and (c) Small World (SW)  (Fig. \ref{figure:leveragegeneratednetworks}) \cite{watts1998collective}.

We find that for ER networks, at very low density (edge probability), the leverage is very low because most edges connect nodes that have a degree of 1. As density increases, we obtain more variation in degrees, and local leverage increases. However, beyond an edge probability of $p=0.05$, leverage decreases as the density of the network decreases. Local leverage thus forms a non-monotonic pattern with ER networks. For SF networks, rather than density or edge probability, we initially examine leverage as  the network becomes more centralized (as $\gamma$ increases above 1, very high degree nodes have a lower probability of occurring). We find that as $\gamma$ increases from 1 to 2, the leverage increases, but then decreases beyond 2. For SW networks, unlike in the ER and SF networks, leverage is monotonically decreasing with number of neighbors (or density), and is monotonically increasing with rewiring probability. In addition, in Figure \ref{figure:leveragegeneratedsize} (\S \ref{section:leveragenetworksize2}), we show how leverage varies with network size, and find that larger networks typically obtain higher leverage for scale-free (SF).

\renewcommand{\baselinestretch}{1} 
\begin{figure}
\begin{centering}
\includegraphics[width=0.9\textwidth]{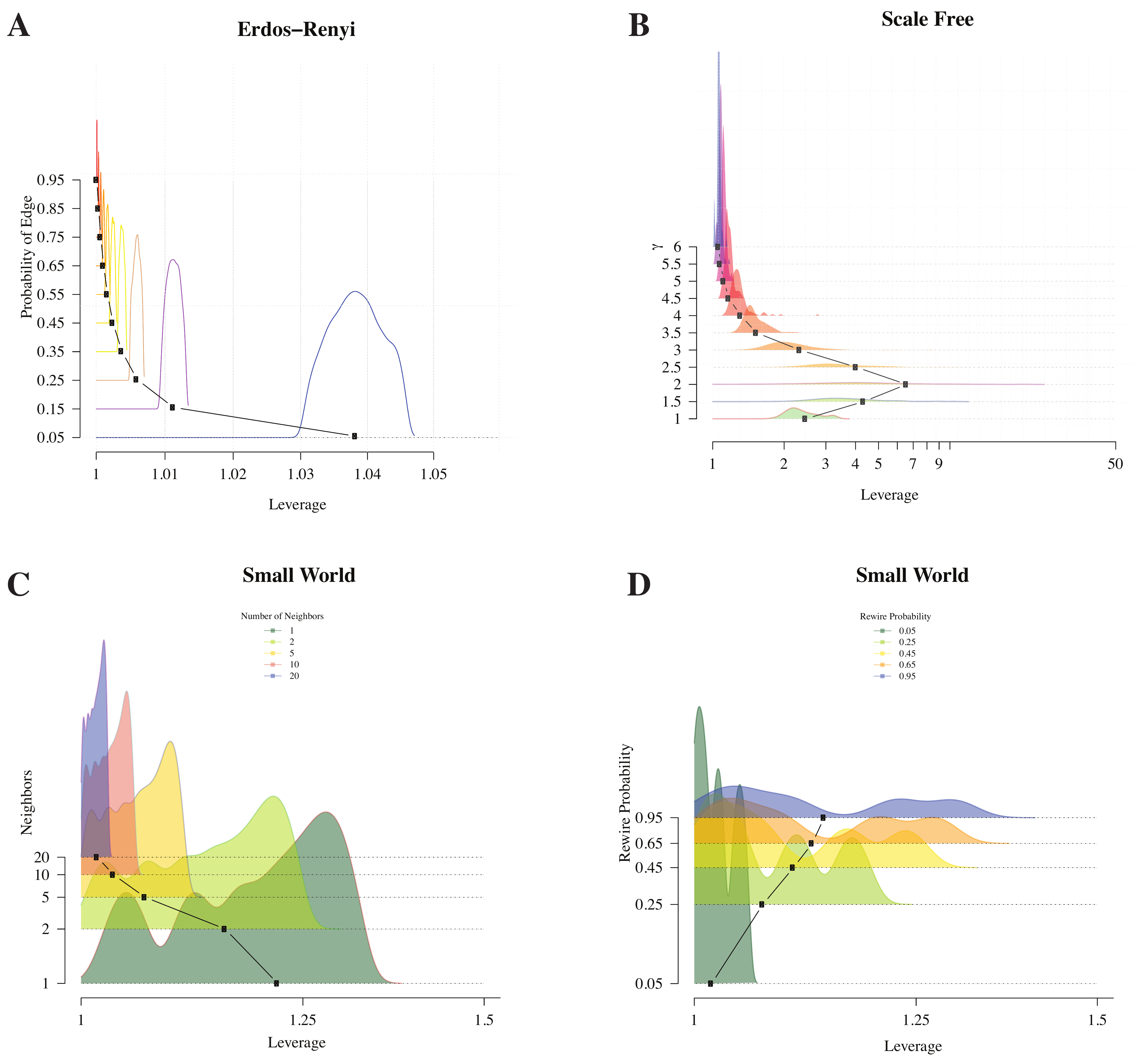}
\par\end{centering}
\caption{\fcaptionsize{
Local Leverage Density in Generated Networks.
To examine whether different generative models result in more or less leverage for the friendship paradox strategies, we examine networks from 3 generative processes.
A sample of 1,000 networks was used for each of the models. (A) Erdos-Renyi (ER) networks generated with edge probabilities, $p \in [0.05,0.95]$, and size ranging from N=50 to N=1000 nodes. We find that local leverage is highest for the lowest edge probabilities, and  leverage converges to 1 as the networks become more dense. \textbf{Overall, ER networks do not achieve high leverage with local and global strategies.}
(B) Static Scale Free (SF) or Barabasi-Albert networks with scale-free parameter $\gamma \in [1,6]$. For these networks, observe that the leverage spans a wider range, e.g. for $\gamma=2$, the samples range from leverage of 1 to over 40. The mean leverage is non-monotonic in terms of $\gamma$, increasing when $\gamma<2$ and decreasing for $\gamma>2$. The distribution of leverage across the samples also displays decreasing variance when $\gamma>2$. At very high levels of $\gamma \approx 6$, the local mean converges to the mean degree. \textbf{Overall, we find that SF networks do achieve high leverage with local and global strategies, and intermediate levels of the gamma parameter obtain highest leverage.}
With small world (SW) or Watts-Strogatz networks, we have two parameters. First is the number of neighbors each node is connected to initially, $n$. The edges are then rewired with a specified probability, $p_r$. First, in panel (C), we find that with a small number of neighbors, the leverage distribution is quite spread out, and there is a substantial leverage effect. However, as we begin to create very dense networks, both the mean and the variance of the leverage distribution leverage diminish substantially. Second, we examine the impact of rewiring probability on the leverage distribution in panel (D). We find that with lower rewiring probabilities, say $p_r=0.05$, the leverage distribution is closer to $1$, whereas with a higher rewiring probabilities, the distributions feature increased variance as well as higher mean leverage.
\textbf{Overall, SW networks result in moderate levels of leverage for local and global strategies.}}}
\label{figure:leveragegeneratednetworks}
\end{figure}
\renewcommand{\baselinestretch}{2}

\paragraph{Real Networks:} \label{sub:strategies-real} We examine the range of real networks detailed in \S \ref{ssection:realnetworks}.
First, observing the local strategy (Fig. \ref{figure:realleverage}A), we find that for all networks, as expected, the friendship paradox strategies are at least as good as the random strategy. Second, for networks like  Twitter (OS4)  or Internet Topology (C1), the leverage can be as high as 100. Thus,  obtaining a friend of a random node will provide a 100-fold increase in the expected degree of a chosen node. Third, we observe that both local and global leverage (Figs. \ref{figure:realleverage}A and \ref{figure:realleverage}B) are higher for nodes when average degree is intermediate, i.e. not too low or high. Some networks like the CA Roads network (I3) have very little degree variation and local and global strategies are relatively less effective. Finally, we examine when local and global strategies make a relative difference (Fig. \ref{figure:realleverage}B). We find that the highest ratio of local to global mean is for Twitter network (OS4), whereas the lowest ratio (indicating that global strategy has a higher expected mean degree) is shown by  Flickr (OS2), both of which belong to the same category of online social networks. Citation networks tend to have higher global mean, whereas for Infrastructure networks, both strategies seem to work just as well.

\section{Application: Controlling Contagion in Networks}\label{section:application}
We demonstrate an application comparing different strategies to control simple contagion spreading through a network. There are a number of models of contagion, and they can be parametrized several ways. However, remarkably most models of contagion can be characterized by a single parameter termed the \textit{epidemic threshold}. If the ratio of infection to that of recovery is lower than the epidemic threshold, then the epidemic is contained and will die out, whereas if the ratio is above the threshold, then it could turn into an epidemic. The epidemic threshold is shown to be a function of both the network and the \textit{virus propagation model} (VPM). The epidemic threshold  of a network is characterized as the inverse of the greatest (first) eigenvalue of the adjacency matrix $A$ of the network, denoted as below (details in \S \ref{ssection:vpm}): 
$$\boxed{\tau(A) = \dfrac{1}{\lambda_1(A)}}$$.

For virtually any VPM, networks with higher epidemic thresholds are less likely to have an epidemic outbreak. The threshold in an undirected network is shown to be proportional to the \textit{inverse of the largest eigenvalue} for a wide range of VPMs, including SIR, SEIR, etc. models that have been commonly used for modeling infectious diseases \cite{prakash2010got}. Nodes are selected for immunization or treatment using each of the intervention strategies (random, local and global). We then examine how the epidemic threshold changes as a function of the proportion of nodes vaccinated (removed), for each strategy.

We examine data in the India villages networks  from \cite{Banerjee1236498}, who collected detailed full census data on the social networks of 75 villages in southern India. The social networks are captured at two different levels of aggregation, at the level of individuals and of households.  %
Details of the network dataset are provided in \S \ref{ssection:realnetworks}.

We find that networks can have either positive or negative inversity depending on how nodes and edges are defined. When nodes as defined as individuals, we find that the networks have negative inversity, whereas if the nodes are defined as households, the inversity values of the resulting networks are mostly positive (Fig. \ref{figure:inversityindiavillages}). Thus, a household-based intervention might use the local strategy, and the individual-based intervention might use the global strategy.

Figure \ref{figure:inversityindiavillages} illustrates the inversity values across the 75 villages separately for individual and household networks. Overall, we find that networks obtained from similar underlying relationships can result in dramatically different inversity characteristics, implying different interventions (local or global) would be better suited. 
\renewcommand{\baselinestretch}{1} 
\begin{figure}
\begin{centering}
\includegraphics[scale=0.6]{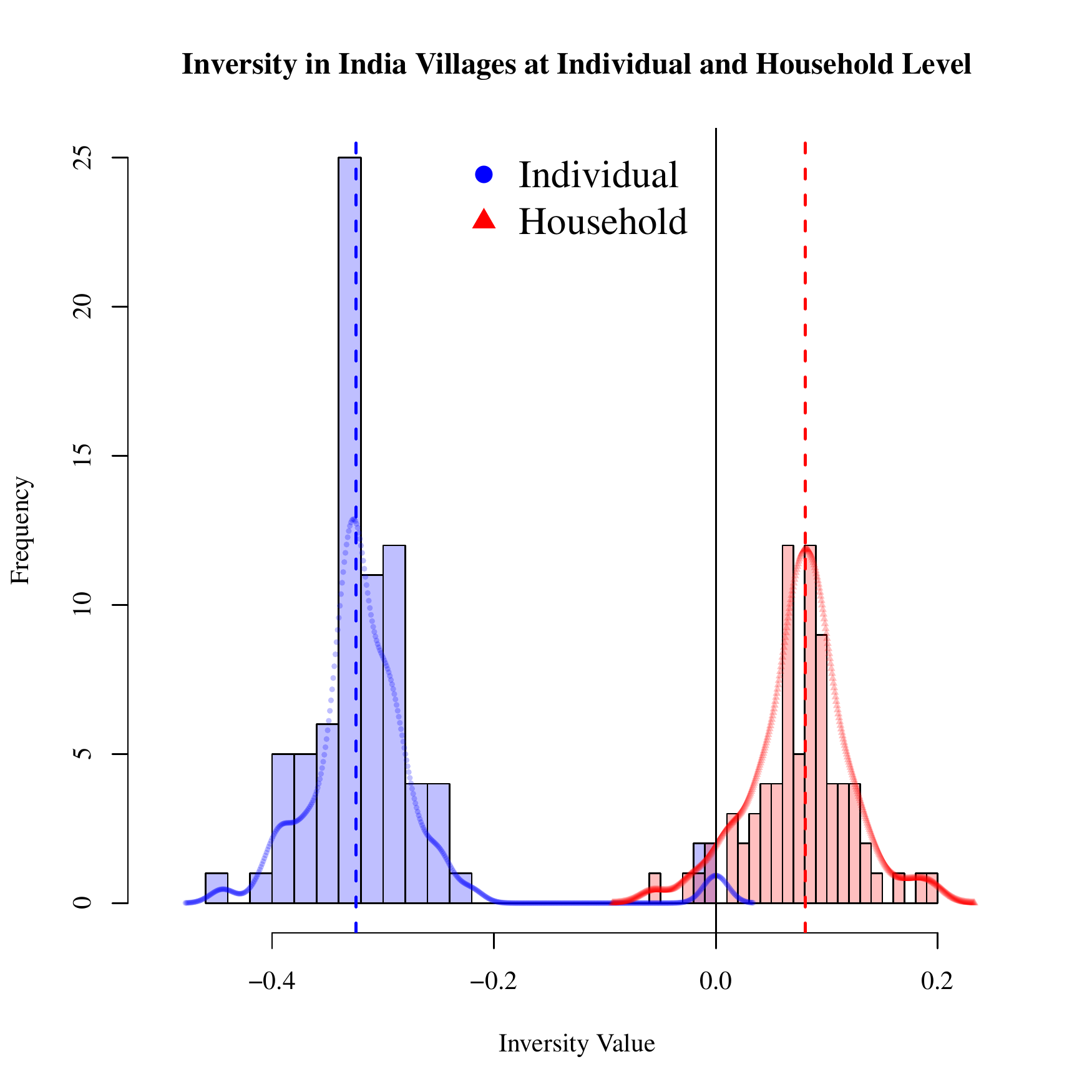}
\par\end{centering}
\caption{Inversity in India Village Networks. Inversity values depend strongly on how the network structure is aggregated. \textbf{We observe  negative inversity values across most of the networks when considering individual-to-individual ties, but  positive inversity values when we consider household-to-household ties.} Note that the household-level ties are aggregated from the individual-level ties.}
\label{figure:inversityindiavillages}
\end{figure}
\renewcommand{\baselinestretch}{2}

\paragraph{Epidemic Threshold and Immunization Strategies:}
Our first goal is to identify how the epidemic threshold $\tau$ changes as we immunize nodes from the network $\mathcal{G}$. While immunizing (or removing)  any node from the network is likely to increase the epidemic threshold, immunizing well-connected nodes is likely to prove especially beneficial. We examine the effectiveness of the three strategies (Random, Local and Global) in identifying which nodes to immunize from the network.

We evaluate the impact of the immunization strategies on the epidemic threshold of a number of real network datasets.
First, we examine the data from $N=75$ village social networks in India  (see \cite{Banerjee1236498}). This dataset is especially useful in our analysis since the villages are relatively isolated, implying they can be evaluated separately. %

In Figure \ref{figure:epidemicthresholdvillages}, we evaluate the eigen threshold ($\tau$) for the networks in the Indian villages data set. First, for the household networks, we find that the local and global strategies obtain a significantly higher epidemic threshold. The difference in threshold between the random strategies and the friendship paradox strategies increases with the proportion of nodes immunized. For individual networks, a similar pattern obtains, but here we find that the global strategy obtains the highest epidemic threshold across all immunization levels, and the difference in thresholds between local and global strategies also increases with the proportion of immunized nodes. This broadly signifies that it is helpful to know which among the global or local strategies to use, and the sign of inversity helps us in making this decision.

\newpage
\renewcommand{\baselinestretch}{1} 
\begin{sidewaysfigure}[!htb]
\includegraphics[width=0.5\textwidth]{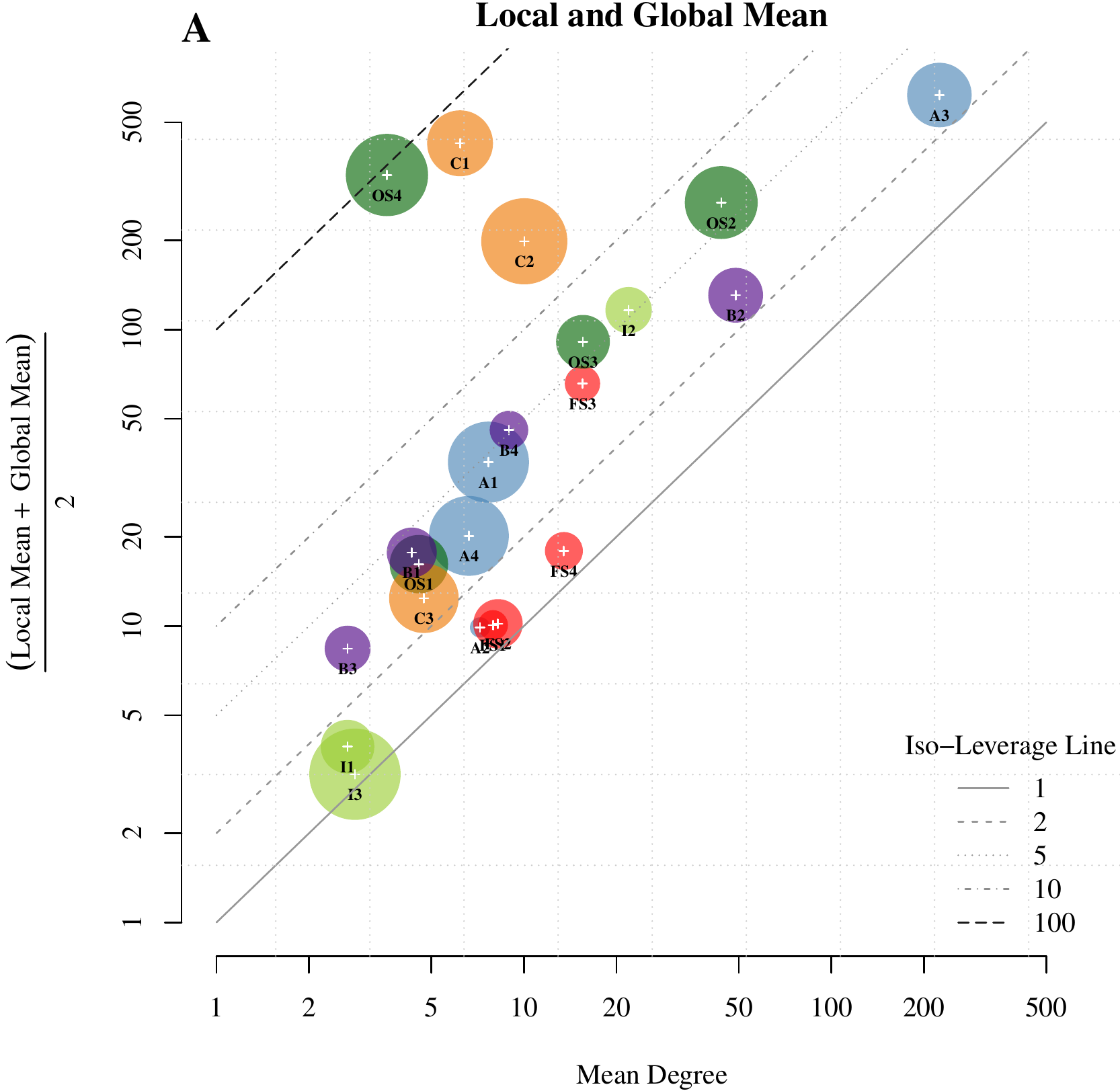} 
\includegraphics[width=0.5\textwidth]{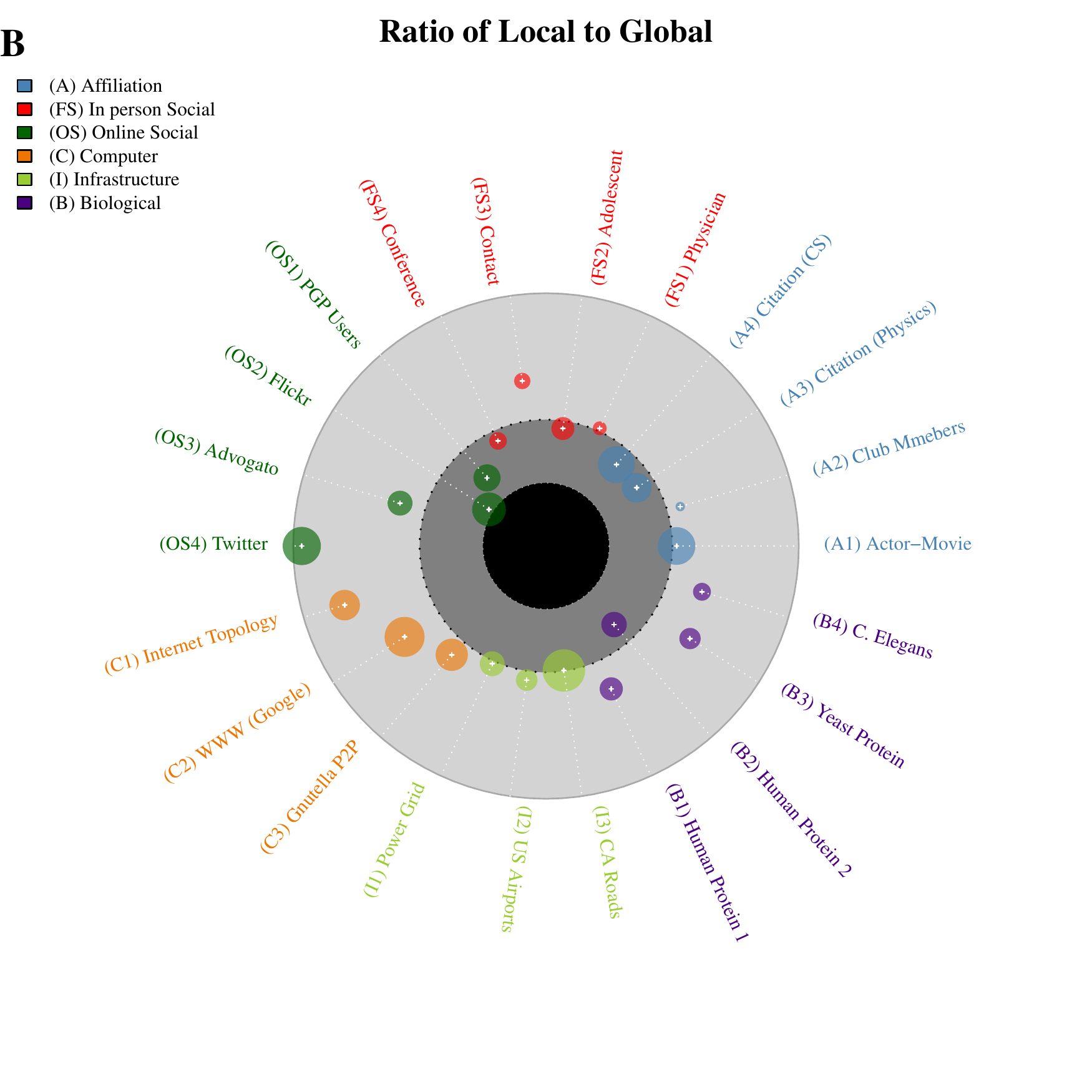}
\caption{
\fcaptionsize{
Global and Local Leverage in Real Networks}  Local and Global Means across Networks (each circle is a network).
Area of circles indicates size of networks (number of nodes) in log scale. Color of circle indicates network category.
(A) The average of \textbf{Local and Global Mean}  is higher than mean degree in all real networks, with the highest differences occurring in online social networks and computer networks. Most large networks also tend to show a higher leverage ratio. For in person or face to face networks, the pattern is more variable. The iso-leverage line indicates leverage levels of 1,2,5,10 and 100. We find that all networks have leverage greater than 1, a majority of networks have leverage greater than 5, and 2 networks have leverage close to 100.
(B) \textbf{Comparison}: Ratio of Local to Global Mean. The ratio of local to global mean $\frac{\mu_L}{\mu_G}$ is represented as follows ($<\frac{1}{2}$ in black circle, $\frac{1}{2}<\frac{\mu_L}{\mu_G}<1$ in dark gray circle and $1<\frac{\mu_L}{\mu_G}<2$ in light gray circle. For example, in the Twitter network, local mean is almost twice the global mean, whereas in the Flickr network, global mean is almost twice the local mean. Computer networks have higher values of the ratio, whereas Infrastructure networks have similar values of local and global means.}
\label{figure:realleverage}
\end{sidewaysfigure}
\clearpage

\renewcommand{\baselinestretch}{1}
\begin{figure}[H]
\includegraphics[scale=1]{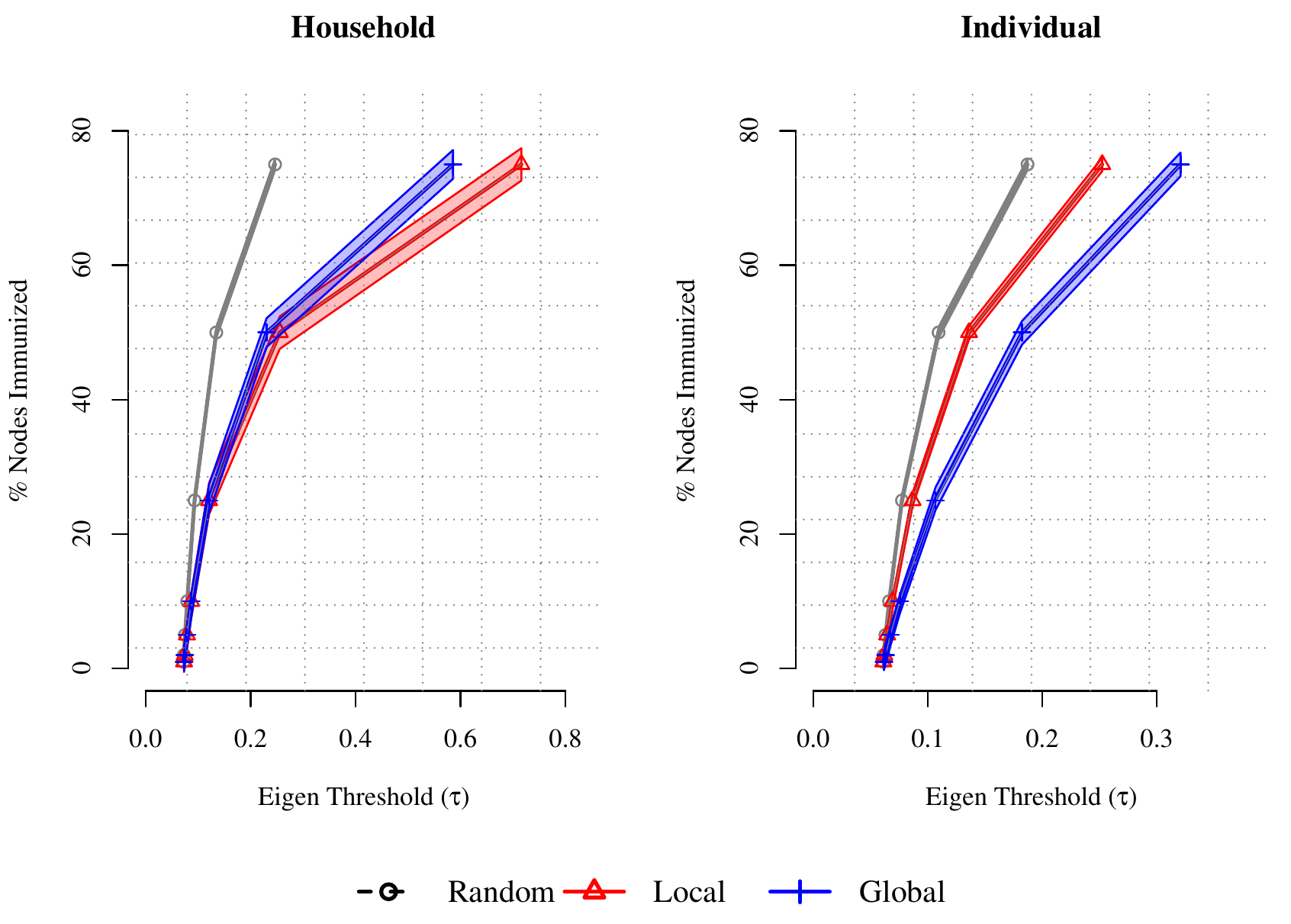}
\begin{centering}
\caption{Epidemic Thresholds with Immunization in India Village Networks. Higher thresholds imply an outbreak is more likely to die out. The dark lines represent the mean values, and the shaded regions are the 95\% confidence intervals. We examine 3 strategies (Random, Local and Global) to choose nodes to immunize. The proportion of nodes immunized ranges from 1\% - 75\%. In both household and individual networks, we find that the friendship paradox strategies obtain higher thresholds than random, for the same proportion of nodes immunized. \textbf{For instance, in the household networks, to achieve a threshold $\tau=0.15$, the random strategy needs to have about 50\% of nodes immunized, but the local and global strategies require less than half of that, at around 25\%.}
For the household networks (left panel), we find that the local strategy is better than the global strategies especially at higher levels of removal. However, for individual networks, we find that the global strategy obtains greater thresholds than local.
}
\label{figure:epidemicthresholdvillages}
\end{centering}
\end{figure}
\renewcommand{\baselinestretch}{2}
\clearpage

\paragraph{Epidemic Outcomes and Immunization Strategies:}
We deploy the epidemic propagation models on 75 village networks from India using an SIR virus propagation model (details in \S \ref{ssection:vpmimplementation}). Define $I_{it}\in\{0,1\}$ as an indicator of whether an individual $i$ is infected at time $t$. We evaluate epidemics on the following aspects:
\begin{itemize}
\item \textbf{Proportion Infected at Peak} =$ \frac{1}{N}\max_t(\sum_i I_{it})$: Since epidemics increase in intensity and eventually die down, an important characteristic is to measure the proportion of the population who are infected at the peak of the epidemic. This directly impacts important decisions like hospital capacity planning etc.
\item \textbf{Proportion Ever Infected} =  $ \frac{1}{N}\sum_i \max_t(I_{it})$: The proportion of the population that was ever infected by the disease is important since  it represents the total spread of the disease in the population. It could also represent the number of people who might have immunity to future recurrences of the disease.%
\item \textbf{Total Suffering}: $ \frac{1}{N T}\sum_i \sum_t(I_{it})$ Here, the total suffering metric captures not just how many infections occur, but also the length of the infections. This represents the proportion of individual-period combinations with an infection.
\end{itemize}

In Figure \ref{figure:epidemicindiavillages}, we evaluate epidemic outcomes using the networks of Indian villages. First, for both household and individual level networks, we find that  strategies based on the friendship paradox, i.e. the global and local strategies perform better than the random strategy. Second, for household networks, the local strategy performs relatively better than the global strategy for each of the epidemic characteristics detailed above. In contrast, for individual-level village networks, we find that the global strategy diminishes the severity of epidemic spread as measured by each of the above characteristics to a greater extent. 
Thus, while it may be beneficial to use either strategy, understanding the role of inversity (as in Figure \ref{figure:inversityindiavillages}) helps determine which of the friendship paradox strategies, i.e. local or global ould result in better epidemic outcomes.
We also examine the outcomes for a network of Facebook users, and find that the local strategy achieves better outcomes on all of the above metrics (see Fig. \ref{figure:epidemicfb} in \S \ref{ssection:fbepidemic}).

\renewcommand{\baselinestretch}{1}
\begin{figure}
\begin{centering}
\includegraphics[scale=0.8]{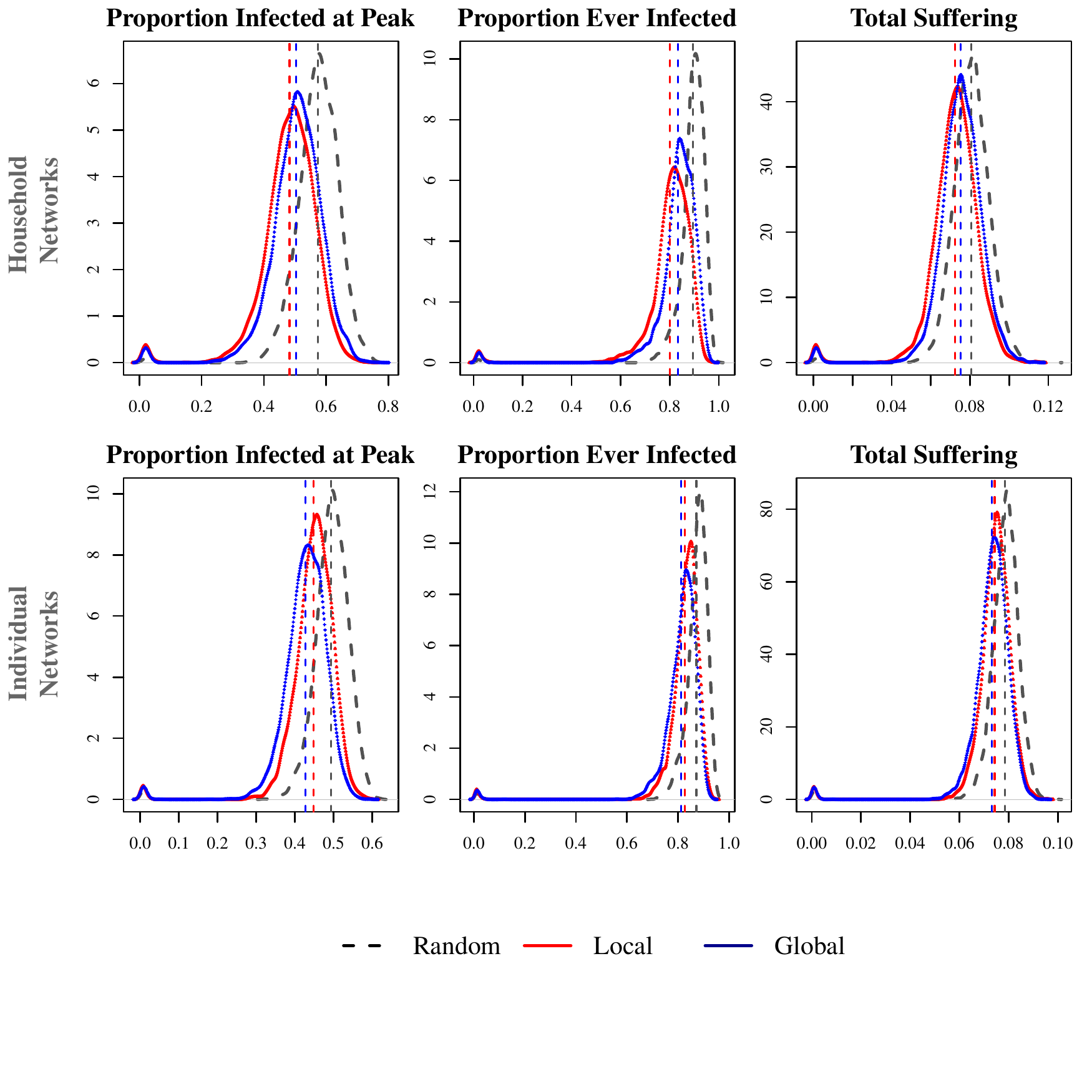}
\vspace{-0.7in}
\end{centering}
\begin{centering}
\caption{Epidemic Characteristics with Immunization in India Village Networks. See Table {table:vpmsimulation} for parameters of simulation. The top 3 panels represent outcomes for household networks, and the bottom 3 panels for individual-level networks. All outcomes are density plots. We plot 3 outcomes: (a) the proportion of population infected at the peak, (b) proportion of population that was ever infected, and (c) total suffering. The $x$-axis represent proportions and the $y$-axis represent density.  We plot the outcomes for 3 strategies: (R)andom, (L)ocal and (G)lobal. The dashed vertical lines represent the means for the 3 strategies. A strategy with a density plot  to the left of another is ``better'' in terms of reducing the severity of the epidemic. \textbf{Thus, for household networks, the local strategy (in red) is better than the global, which in turn is better than the random strategy. This ordering is the same for all 3 outcomes. For the individual networks, however, the global strategy is ``bettter,''  for all 3 outcomes.}}
\label{figure:epidemicindiavillages}
\end{centering}
\end{figure}
\renewcommand{\baselinestretch}{2}
\clearpage

\section{Conclusion}\label{section:conclusion}
We show that with unknown networks, the friendship paradox can be leveraged to obtain such individuals with minimal informational requirements. We identify intervention strategies (local and global) that have theoretical guarantees on obtaining better-connected individuals. With both generated random networks and real networks, our results show the value of using the local and global strategies to obtain highly connected nodes. In the vast majority of networks, we obtain at least double the average degree, and some networks show increases of several hundred-folds in node degree. We expect the advantages of speed of implementation, generality of application areas for these privacy-sensitive and informationally-light strategies to provide an important tool for network interventions in unknown structures.

\bibliography{structuralleverage}
\bibliographystyle{unsrt}

\clearpage

\newpage

\newcommand{\beginsupplement}{%
        \setcounter{table}{0}
        \renewcommand{\thetable}{S\arabic{table}}%
        \setcounter{figure}{0}
        \renewcommand{\thefigure}{S\arabic{figure}}%
	\setcounter{section}{0}
        \renewcommand{\thesection}{S.\Alph{section}}%
	\setcounter{theorem}{0}
        \renewcommand{\thetheorem}{S\arabic{theorem}}%
	\setcounter{lemma}{0}
        \renewcommand{\thelemma}{S\arabic{lemma}}%
}
\beginsupplement
\renewcommand{\baselinestretch}{2} 

\setcounter{page}{1}
\section*{Supplementary Material}

\section{Mathematical Appendix}\label{section:mathproofs}
Formally, the network graph $\mathcal{G}=\left(V,E\right)$ is comprised of
a set of $N$ individual nodes and a set of undirected edges $E$.
Each element of $E$ is a pair of nodes, $(i,j)$ indicates an edge
(connection) with $e_{ij}\in\left\{ 0,1\right\} $. We also define
the directed edge set $\hat{E}$ including both $(i,j)$ and $(j,i)$
as distinct elements of $\hat{E}$ corresponding to an undirected
edge $i\leftrightarrow j$. We detail the table of notation in Table \ref{table:notation-1}.
\begin{table}[h]
\caption{Table of Notation}
\label{table:notation-1}%
\begin{tabular}{ccc}
\hline 
Symbol & Term & Definition \tabularnewline
\hline 
$\mathcal{G}$,$V$,$E$ & Network & Network Graph of Nodes $V$ and Edges $E$\tabularnewline
$\hat{E}$ & Directed Edge Set & Each edge in $E$ is replaced by two directed edges\tabularnewline
$D_{i}$ & Degree & Number of friends of $i$\tabularnewline
$\mathscr{N}(i)$ & Neighbors & \emph{Set} of Friends of $i$\tabularnewline
$F_{i}$ & Average degree of friends of $i$ & $\frac{1}{D_{i}}\sum_{j\in\mathscr{N}(i)}D_{j}$\tabularnewline
$\mu_{D},\sigma_{D}^{2}$ & Mean and variance of Degrees & $\frac{1}{N}\sum_{i}D_{i},\ \frac{1}{N}\sum_{i}\left(D_{i}-\mu_{D}\right)^{2}$\tabularnewline
$\mu_{L}$ & Local Mean & $\frac{1}{N}\sum_{i}F_{i}$\tabularnewline
$\mu_{G}$ & Global Mean & $\frac{\sum_{i}D_{i}F_{i}}{\sum_{i}D_{i}}$\tabularnewline
$\rho$ & Inversity & $Corr\left(D_{i},\frac{1}{D_{j}}\right)\forall(i,j)\in\hat{E}$\tabularnewline
\hline 
\end{tabular}
\end{table}

The basic idea of the friendship paradox can be expressed as "your friends have more friends than you." We examine the degree to which the friendship paradox holds for individual nodes, or the individual friendship paradox. We find in the result below that it cannot hold for all nodes, but can hold for an arbitrarily high proportion ($<1$) of nodes.
\begin{theorem}
\label{theorem:individualfp}
The friendship paradox statement that "your friends have more friends than you" cannot hold for all nodes in a network. Also, the statement can hold for all nodes, except one.
\end{theorem}
\begin{proof}
Consider a connected network where not all degrees are identical (if all are identical, the statement cannot hold). There must be at least one node that has the highest degree $D_{max}$ and which is connected to at least one node with a lower degree. If not, then the connected network is comprised entirely of highest (identical) degree nodes, thus contradicting the initial statement. If the highest degree node is connected to a lower degree node, then the average friends of friends of the highest degree node must be lower than $D_{max}$. Thus the statement cannot hold for \textit{all} nodes. To show the second part that it can hold for all nodes except one, consider the star  (hub and spoke) network, where all of the nodes except the central node have fewer friends than their friends do.
\end{proof}

\begin{theorem}\label{theorem:globalmean}
{[}Feld 1991{]} For a network $\mathcal{G}=\left(V,E\right)$ with degree mean
$\mu_{D}$ and variance $\sigma_{D}^{2}$, the global mean of friends
of friends is 
\begin{equation}
\mu_{G}=\left(\mu_{D}+\frac{\sigma_{D}^{2}}{\mu_{D}}\right)
\end{equation}
\end{theorem}

\begin{proof}\emph{(as given in Feld, 1991).} $\mu_{G}=\dfrac{\sum_{i}\sum_{j}e_{ij}D_{j}}{\sum_{i}D_{i}}=\dfrac{\sum_{i}D_{i}^{2}}{\sum_{i}D_{i}}=\dfrac{\mu_{D}^{2}+\sigma_{D}^{2}}{\mu_{D}}$$\ $
\end{proof}

\begin{theorem}\label{theorem:localmean}
For any general network $\mathcal{G}=(V,E)$ with mean degree $\mu_{D}$, the
local mean of friends is given by 
\begin{equation}
\mu_{L}=\mu_{D}+\frac{1}{2|V|}\sum_{(i,j)\in V\times V}e_{ij}\left[\frac{\left(D_{i}-D_{j}\right)^{2}}{D_{i}D_{j}}\right]
\end{equation}
where $D_{i}$ is the degree of node $i$, and $e_{ij}\in\left\{ 0,1\right\} $
indicates a connection between $i$ and $j$.
\end{theorem}

\begin{proof} Let $D_{i}$ denote the number of connections
of individual $i$, i.e. $D_{i}=\big|\left\{ k\in V:\quad(i,k)\in E\right\} \big|$.
Denote the set of neighbors of $i$ by $\mathscr{N}(i)=\left\{ k\in V:\quad(i,k)\in E\right\} $.
Define $F_{i}=\frac{1}{D_{i}}\sum_{j\in N(i)}D_{j}$ as the mean number
of friends for friends of $i$. The local mean is defined as: 
\[
\mu_{L}=\frac{1}{|V|}\sum_{i}F_{i}=\sum_{i\in V}\left[\frac{1}{D_{i}}\left(\sum_{j\in\mathscr{N}(i)}D_{j}\right)\right]
\]

Rewriting the expression for $\mu_{L}$ in terms of the connections
(edges) between individuals, we obtain:
\begin{eqnarray*}
\mu_{L} & = & \frac{1}{|V|}\sum_{i\in V}\left[\frac{1}{D_{i}}\left(\sum_{j\in V}e_{ij}D_{j}\right)\right]=\frac{1}{|V|}\sum_{i\in V}\sum_{j\in V}\left[e_{ij}\frac{1}{D_{i}}\left(D_{j}\right)\right]\\
 & = & \frac{1}{2|V|}\sum_{(i,j)\in V\times V}\left[e_{ij}\left(\frac{D_{j}}{D_{i}}\right)+e_{ji}\left(\frac{D_{i}}{D_{j}}\right)\right]=\frac{1}{2|V|}\sum_{(i,j)\in V\times V}e_{ij}\left[\frac{D_{j}}{D_{i}}+\frac{D_{i}}{D_{j}}\right]\\
 & = & \frac{1}{2|V|}\sum_{(i,j)\in V\times V}e_{ij}\left[\frac{D_{j}^{2}+D_{i}^{2}}{D_{i}D_{j}}\right]=\frac{1}{2|V|}\sum_{(i,j)\in V\times V}e_{ij}\left[\frac{\left(D_{i}-D_{j}\right)^{2}+2D_{i}D_{j}}{D_{i}D_{j}}\right]\\
 & = & \frac{1}{2|V|}\sum_{(i,j)\in V\times V}e_{ij}\left[\frac{\left(D_{i}-D_{j}\right)^{2}}{D_{i}D_{j}}\right]+\frac{1}{2|V|}\left(4|E|\right)\\
 & = & \mu_{D}+\frac{1}{2|V|}\sum_{(i,j)\in V\times V}e_{ij}\left[\frac{\left(D_{i}-D_{j}\right)^{2}}{D_{i}D_{j}}\right]\qed
\end{eqnarray*}
\end{proof}
Note that what we characterize as the local mean defined as above was examined by  others including \cite{Feld1991} etc. and was independently shown to be greater than the mean degree by us (\cite{krackhardt2002}) and others (including by Christian Borgs \& Jennifer Chayes 
in an online comment to an article by \cite{Strogatz2012}, and by \cite{jackson2016friendship}). However, the properties of the local mean htave not been formally examined and characterized.

\begin{theorem}\label{theorem:inversity}
Define the m-th moment of the degree distribution by $\kappa_{m}=\frac{1}{N}\sum_{i\in V}D_{i}^{m}$.
The local and global means are connected by the following relationship
involving the inversity $\rho$ and the -1,1,2, and 3rd moments of
the degree distribution as follows:
\begin{eqnarray}
\mu_{L}=\mu_{G}+\rho\sqrt{\left(\frac{\kappa_{1}\kappa_{3}-\kappa_{2}^{2}}{\kappa_{1}}\right)\left[\kappa_{-1}-\left(\kappa_{1}\right)^{-1}\right]}
\end{eqnarray}
\end{theorem}

\begin{proof} Define the moments of the degree distribution
as: $\kappa_{m}=\frac{1}{N}\sum_{i}D_{i}^{m}$. Since we defined $\rho$
the measure of inversity as the correlation of two distributions that
we specify as the origin degree (\textbf{O}) and inverse desitnation
degree (\textbf{ID}) distributions. The \textbf{O} distribution consists
of the degree of nodes corresponding to edges, and \textbf{ID} distribution
consists of the inverse degree of nodes corresponding to edges. Thus,
each connection (edge) contributes \emph{two} entries to \emph{each}
distribution. For example, if there is a connection between $i$ and
$j$, i.e. $e_{ij}=1$, we would have $\left(D_{i},\frac{1}{D_{j}}\right)$
and $\left(D_{j},\frac{1}{D_{i}}\right)$. Observe that each individual
appears in both distributions multiple times based on degree.

Next, we detail the mean and variance of the distributions. First,
we consider the means. The mean of the origin distribution is $\mu_{O}=\frac{1}{2|E|}\sum_{i}D_{i}^{2}=\frac{\mu_{D}^{2}+\sigma_{D}^{2}}{\mu_{D}}=\mu_{G}=\frac{\kappa_{2}}{\kappa_{1}}$.
Similarly, the \textbf{ID} mean is $\mu_{ID}=\frac{1}{2|E|}\sum_{i}D_{i}\left(\frac{1}{D_{i}}\right)=\frac{1}{\mu_{D}}$.
Next, consider the variances. The variance of the origin distribution
(\textbf{O}) is computed as:
\begin{eqnarray*}
\sigma_{O}^{2} & = & \frac{1}{2|E|}\sum_{\left(i,j\right)\in E}\left(D_{i}-\mu_{O}\right)^{2}=\frac{1}{2|E|}\sum_{i\in V}D_{i}\left(D_{i}-\mu_{O}\right)^{2}\\
 & = & \frac{1}{N\mu_{D}}\sum_{i\in V}\left[D_{i}^{3}-2\mu_{O}D_{i}^{2}+\left(\mu_{O}\right)^{2}D_{i}\right]=\frac{\kappa_{3}}{\kappa_{1}}-\left(\frac{\kappa_{2}}{\kappa_{1}}\right)^{2}
\end{eqnarray*}
Next, we express the corresponding variance of the inverse destination
degree distribution (\textbf{ID}), $\sigma_{ID}^{2}$. Again, recall
that $\frac{1}{D_{i}}$ does not appear just once, but $D_{i}$ times.
Therefore, we have:
\begin{eqnarray*}
\sigma_{ID}^{2} & = & \frac{1}{2|E|}\sum_{(i,j)\in E}\left[\left(\frac{1}{D_{j}}-\frac{1}{\mu_{D}}\right)^{2}\right]=\frac{1}{2|E|}\sum_{(i,j)\in E}\left(\frac{1}{D_{j}^{2}}+\frac{1}{\mu_{D}^{2}}-\frac{2}{\mu_{D}D_{j}}\right)\\
 & = & \frac{1}{2|E|}\left[\sum_{(i,j)\in E}\frac{1}{D_{j}^{2}}+\frac{1}{\mu_{D}^{2}}\left(\sum_{(i,j)\in E}1\right)-\frac{2}{\mu_{D}}\sum_{(i,j)\in E}\frac{1}{D_{j}}\right]=\frac{1}{2|E|}\left[\sum_{j\in V}\frac{1}{D_{j}}+\frac{1}{\mu_{D}^{2}}2|E|-\frac{2}{\mu_{D}}N\right]\\
 & = & \frac{1}{\mu_{D}N}\left[\sum_{j\in V}\frac{1}{D_{j}}\right]-\frac{1}{\mu_{D}^{2}}\qquad=\left(\kappa_{1}\right)^{-1}\left[\kappa_{-1}-\left(\kappa_{1}\right)^{-1}\right]
\end{eqnarray*}
We next turn to the inversity and based on the definition we connect
it to the local and global means and the degree distribution.
\begin{eqnarray*}
\rho & = & \left(\frac{1}{2|E|\sigma_{O}\sigma_{ID}}\right)\sum_{\left(i,j\right)\in E}e_{ij}\left[\left(D_{i}-\mu_{O}\right)\left(\frac{1}{D_{j}}-\frac{1}{\mu_{D}}\right)\right]\\
\left(N\mu_{D}\sigma_{O}\sigma_{ID}\right)\rho & = & \left[\sum_{\left(i,j\right)\in E}e_{ij}\left(\frac{D_{i}}{D_{j}}\right)-\mu_{O}\left(\sum_{\left(i,j\right)\in E}\frac{1}{D_{j}}\right)-\frac{1}{\mu_{D}}\sum_{\left(i,j\right)\in E}D_{i}+\sum_{\left(i,j\right)\in E}e_{ij}\left(\frac{\mu_{O}}{\mu_{D}}\right)\right]\\
 & = & \left[N\left(\mu_{L}\right)-\mu_{O}\cdot N-\frac{1}{\mu_D}\sum_{\left(i,j\right)\in E}D_{i}+\sum_{\left(i,j\right)\in E}e_{ij}\left(\frac{\mu_{O}}{\mu_{D}}\right)\right]\\
 & = & \left[\left(N\mu_{L}\right)-N\mu_{O}-\frac{1}{\mu_D}\sum_{\left(i,j\right)\in E}D_{i}+2|E|\left(\frac{\mu_{O}}{\mu_{D}}\right)\right]\\
\implies\mu_{L} & = & \mu_{G}+\rho  \cdot  \mu_D \cdot \ \sigma_{O}\sigma_{ID}
\end{eqnarray*}
Finally, substituting $\mu_D = \kappa_1$ and the expressions for the variances, we obtain:
\begin{equation}
\mu_{L}=\mu_{G}+\rho\sqrt{\left(\frac{\kappa_{1}\kappa_{3}-\kappa_{2}^{2}}{\kappa_{1}}\right)\left[\kappa_{-1}-\left(\kappa_{1}\right)^{-1}\right]}
\end{equation}
\end{proof}

\begin{theorem}\label{theorem:globalintervention}
The expected degree of nodes chosen by global strategy is the global mean.
\end{theorem}

\begin{proof}
To determine the expected degree of a node chosen by the global strategy:
Choose $M=1$ node initially, (say X). With probability $q$, choose
each neighbor of X. For a node $k$ with degree $D_{k}$, the probability
of being chosen by this process is the first step when any of $k$'s friends is chosen as the initial node, and the second step is $k$ being chosen with probability $q$. This probability is $p_{k}=\frac{1}{N}\ D_{k} \times q = \frac{q D_k}{N}$.
The expected degree
of a chosen ``seed'' node is then the degree-weighted probability:
\[
\frac{\sum_{k\in V}\ p_{k}D_{k}}{\sum_{k\in V}p_{k}}=\frac{\sum_{k\in V}\frac{1}{N}q\ D_{k}^{2}}{\sum_{k\in V}\frac{1}{N}q\ D_{k}}=\frac{\frac{1}{N}\sum_{k\in V}D_{k}^{2}}{\frac{1}{N}\sum_{k\in V}D_{k}}=\frac{\mu_D^2+\sigma_D^2}{\mu_D}=\mu_{G}
\]
\end{proof}
Similar logic applies if we choose any arbitrary initial sample of
size $M$ as long as the network is large, i.e. $N\gg M$.

\begin{obs*}
Denote an undirected tie $(a,b)$ as a connection between nodes $a$ and $b$.
For any network with a given distribution of degrees, the distribution
of degrees is unchanged if any two ties $(a,b)$ and $(c,d)$ are rewired to either (i) $(a,c)$,$(b,d)$
or (ii) $(a,d)$,$(b,c)$.
\end{obs*}
Each of these nodes loses one tie and gains another and therefore
the degrees are unchanged.

\begin{theorem}\label{theorem:rewire} [Rewiring Theorem]
Let network $\mathcal{G}=(V,E)$ with $N>3$ nodes include nodes $a,b,c,d$ with degrees ordered as: $D_{a}\le D_{b}<D_{c}\le D_{d}$. If $G$ containing edges $(a,b),(c,d)\in E$, but $(a,d),(b,c)\notin E$ is rewired to network $\mathcal{G}'=\left(V,E'\right)$, containing edges  $(a,d),(b,c)\in E'$, but  $(a,b),(c,d)\notin E'$, then $\mathcal{G}'$ has higher local mean than $\mathcal{G}$.
\end{theorem}

\begin{proof}
First, observe that the degree distribution is unaffected by the change, and therefore the global mean (which only depends on mean and variance of the
degree distribution) is also unaffected, i.e. $\mu_G(\mathcal{G})=\mu_G(\mathcal{G}')$. Recall that the local mean is $\mu^{L}=\frac{1}{N} \sum_i \sum_j e_{ij} \left[\dfrac{D_i}{D_j}+\dfrac{D_j}{D_i}\right]$. Since between $\mathcal{G}$ and $\mathcal{G}'$ the degrees of all nodes are the same, and all edges are the same except the two rewired edges, we can write the difference between the local means
the local means as:
\begin{eqnarray*}
\mu^{L}(\mathcal{G}')-\mu^{L}(\mathcal{G}) & = &\frac{1}{N}  \left[\left(\frac{D_{a}}{D_{d}}+\frac{D_{d}}{D_{a}}+\frac{D_{b}}{D_{c}}+\frac{D_{c}}{D_{b}}\right)-\left(\frac{D_{a}}{D_{b}}+\frac{D_{b}}{D_{a}}+\frac{D_{c}}{D_{d}}+\frac{D_{d}}{D_{c}}\right)\right]\\
 & = &\frac{1}{N}  \left[ \left(D_{d}-D_{b}\right)\left(\frac{1}{D_{a}}-\frac{1}{D_{c}}\right)+\left(D_{c}-D_{a}\right)\left(\frac{1}{D_{b}}-\frac{1}{D_{d}}\right)\right]
>0
\end{eqnarray*}
The last inequality follows from the ordering of the node degrees. Note that we actually only require the conditions $D_b<D_d$ and $D_a<D_c$ to hold.
\end{proof}

\begin{lemma}\label{lemma:localmaxmin}
Given a connected network with $|V|=N>3$ nodes and any non-degenerate
degree distribution. To achieve maximum local mean among all networks satisfying the given degree distribution, the nodes with maximum
and minimum degree must be connected to each other.
\end{lemma}
\begin{proof}
We prove this by contradiction. Let the network $\mathcal{G}=(V,E)$ have the maximum local mean for the specified degree distribution.
Label $a$ and $z$ as the nodes with minimum and maximum degrees in our
network. These degrees must be different ($D_{a}\ne D_{z}$) in a non-degenerate distribution.
Assume $a$ and $z$ are not connected to each other.

There must be a highest degree node $z$ connected to a node $y$ that satisfies the following conditions:
(1) $y$ is not directly connected to $a$, i.e. $(a,y)\notin E$ and (2) $D_a<D_y<D_z$,
Note that (1) must be satisfied since $D_z \ge D_a$, and (2) must be satisfied since $a$ and $z$ are lowest and highest degree nodes. Now, we can find a neighbor of $a$, say $b$ with $D_b<D_z$. Choose a neighbor $x$ of $z$ that is not connected to $b$. $x$ must exist, otherwise $b$ and $z$ would have the same degree, contradicting the assumption that lowest and highest degree nodes are not connected. Observe that we can increase the local mean by rewiring the network to $\mathcal{G}'$ by connecting $(a,z)$ and $(b,x)$ in place of $(a,b)$ and $(x,z)$ as in Theorem \ref{theorem:rewire} above.
Thus, network $\mathcal{G}$ that we started with could not have had the maximum local mean, and we have a contradiction. Thus, the statement of the theorem must hold.
\end{proof}

\begin{theorem}\label{theorem:leveragemax}
If the degree distribution is unconstrained, the star network maximizes the local leverage $\lambda_L=\frac{\mu_L}{\mu_D}$. 
\end{theorem}

\begin{proof}
Let $\delta$ and $\Delta$ be the minimum and maximum degree. 
Define $f(x,y)=\frac{x}{y}+\frac{y}{x}$.
Without loss of generality, assume that $x\ge y$. 
First, observe that $f(x+1,y-1)>f(x,y)$. To prove this, we can express 
\begin{eqnarray*}
f(x+1,y-1)-f(x,y)
&=&
\frac{x+1}{y-1}+\frac{y-1}{x+1} - \left(\frac{x}{y} + \frac{y}{x}\right)\\
&=&
\left(x+y\right)\left[\frac{1}{y(y-1)}-\frac{1}{x(x+1)}\right]
> 0
\end{eqnarray*}
where the last inequality follows from the assumption $x>y$.
Thus, the maximum value of $f(x,y)$ when $\delta \le x,y\le \Delta$ is at $x=\Delta$, $y=\delta$ or $f(\Delta,\delta)=\frac{\Delta}{\delta}+\frac{\delta}{\Delta}$.
Observe that the ratio of local mean to mean degree can be expressed as 
\begin{eqnarray*}
\lambda_L=\frac{\mu_L}{\mu_D} &=& 1+\frac{ \sum_{i,j \in V} e_{ij} \left[\dfrac{D_i}{D_j}+\dfrac{D_j}{D_i}\right]}{\sum_{i,j}e_{ij}}
\end{eqnarray*}
Thus, for each edge, the maximum value of $\left[\dfrac{D_i}{D_j}+\dfrac{D_j}{D_i}\right]$ from above is bounded by $\left(\frac{\Delta}{\delta}+\frac{\delta}{\Delta}\right)$, and the maximum local leverage is $\lambda_L^{\max}=1+\left(\frac{\Delta}{\delta}+\frac{\delta}{\Delta}\right)$. Observe that expression is maximum when the highest degree node is $\Delta=N-1$ and is connected to a lowest degree node of degree $\delta=1$, which implies a star network. Therefore, no network can have higher local leverage than the star network.
\end{proof}

\newpage
\section{Data on Real Networks}\label{ssection:realnetworks}
We use a wide variety of real networks to determine properties and illustrate of the networks as it relates to the interventions detailed in the paper. We use data from two repositories.

\subsection{Koblenz Network Collection}
The networks are selected across several categories (Affiliation, Face-to-face Social, Online Social, Computer, Infrastructure and Biological networks), and span a wide range in network characteristics like size and density (Table \ref{table:realnetworks2}). These networks also vary widely in terms of their size, from a low of 25 to networks with millions of nodes (e.g. Youtube).
All network data was obtained from the Koblenz Network Collection \cite{kunegis2013konect}. We examine these real networks on a number of dimensions, the number of nodes, edges and the variation in the degree distribution.
\renewcommand{\baselinestretch}{1} 
\begin{table}[H]
\centering{}\caption{Real Network Characteristics}
\label{table:realnetworks2}%
\begin{tabular}{clccccc}
\hline
\hline
\hline
Label & Network Name & Nodes & Edges & Min Degree & Max Degree \\ 
\hline
\multicolumn{6}{l}{\textit{Collaboration}}\tabularnewline
A1 & Actor-Movie & 383640 & 1470338 & 1 & 655 \\ 
  A2 & Club Mmebers &  25 & 91 & 3 & 20 \\ 
  A3 & Citation (Physics) & 28045 & 3148413 & 1 & 4909 \\ 
  A4 & Citation (CS) & 317080 & 1049865 & 1 & 343 \\ 
\multicolumn{6}{l}{\textit{Face-toFace Interaction}}\tabularnewline
  FS1 & Physician & 117 & 464 & 2 & 26 \\ 
  FS2 & Adolescent & 2539 & 10454 & 1 & 27 \\ 
  FS3 & Contact & 274 & 2124 & 1 & 101 \\ 
  FS4 & Conference & 410 & 2765 & 1 & 50 \\ 
\multicolumn{6}{l}{\textit{Online Social}}\tabularnewline
  OS1 & PGP Users & 10679 & 24315 & 1 & 205 \\ 
  OS2 & Flickr & 105722 & 2316667 & 1 & 5425 \\ 
  OS3 & Advogato & 5042 & 40509 & 1 & 803 \\ 
  OS4 & Twitter & 465016 & 833539 & 1 & 677 \\ 
\multicolumn{6}{l}{\textit{Topology of Computer Networks}}\tabularnewline
  C1 & Internet Topology & 34761 & 107719 & 1 & 2760 \\ 
  C2 & WWW (Google) & 855802 & 4291352 & 1 & 6332 \\ 
  C3 & Gnutella P2P & 62561 & 147877 & 1 & 95 \\ 
\multicolumn{6}{l}{\textit{Infrastructure}}\tabularnewline
  I1 & Power Grid & 4941 & 6593 & 1 & 19 \\ 
  I2 & US Airports & 1572 & 17214 & 1 & 314 \\ 
  I3 & CA Roads & 1957027 & 2760387 & 1 & 12 \\ 
\multicolumn{6}{l}{\textit{Biological}}\tabularnewline
  B1 & Human Protein 1 & 2783 & 6222 & 1 & 129 \\ 
  B2 & Human Protein 2 & 5973 & 146385 & 1 & 855 \\ 
  B3 & Yeast Protein & 1458 & 1970 & 1 & 56 \\ 
  B4 & C. Elegans & 453 & 2033 & 1 & 237 \\ 
 \hline
\hline
\hline
\end{tabular}
\end{table}

\subsection{India Village Networks}
In addition, we also use data from $N=75$ villages in India made publicly available (see \cite{Banerjee1236498} for details). The summary statistics for those village household networks are detailed in Table \ref{table:summarystatsnet}.
\begin{table}[H]
\centering
\caption{Summary Statistics of Village Networks}\label{table:summarystatsnet}
\begin{tabular}{l
*{4}{S[table-format=3.2]}
}
\hline
{\textbf{Network Statistic}} & {\textbf{Mean}} & {\textbf{SD}} & {\textbf{Min}} & {\textbf{Max}}  \\
\hline
Number of households & 216.69 & 61.22 & 77 & 356 \\
Number of (undirected) edges & 993.31 & 348.77 & 334 & 2015 \\
Density &0.05   &0.02 &0.02 &0.11 \\
Degree Mean &9.10 &1.573 &6.13 & 12.78 \\
Degree Variance & 52.03 &19.88 &27.80 & 124.56 \\
\hline
\end{tabular}
\end{table}
\renewcommand{\baselinestretch}{2} 

\section{Individual Friendship Paradox}\label{section:sindividualfp}
A basic view of the friendship paradox is developed by plotting the average number of friends (degree) of individual nodes' ``friends'' on the vertical axis against the average degree (Fig. \ref{figure:fpindividual}, Fig. \ref{figure:individualfpcdf}). For example, in the Contact (In person Social) network, we see a deep blue region above and to the left of the $45^{\circ}$ line. Although present across all networks, the pattern is most prominent in the WWW (Google) or Twitter (Online Social) network. Observe also that in the Road Network, only $\Delta=37\%$ of nodes have a higher average number of friends of friends than their own degree.%
\renewcommand{\baselinestretch}{1}
\begin{figure}
\centering{}
\includegraphics[scale=1]{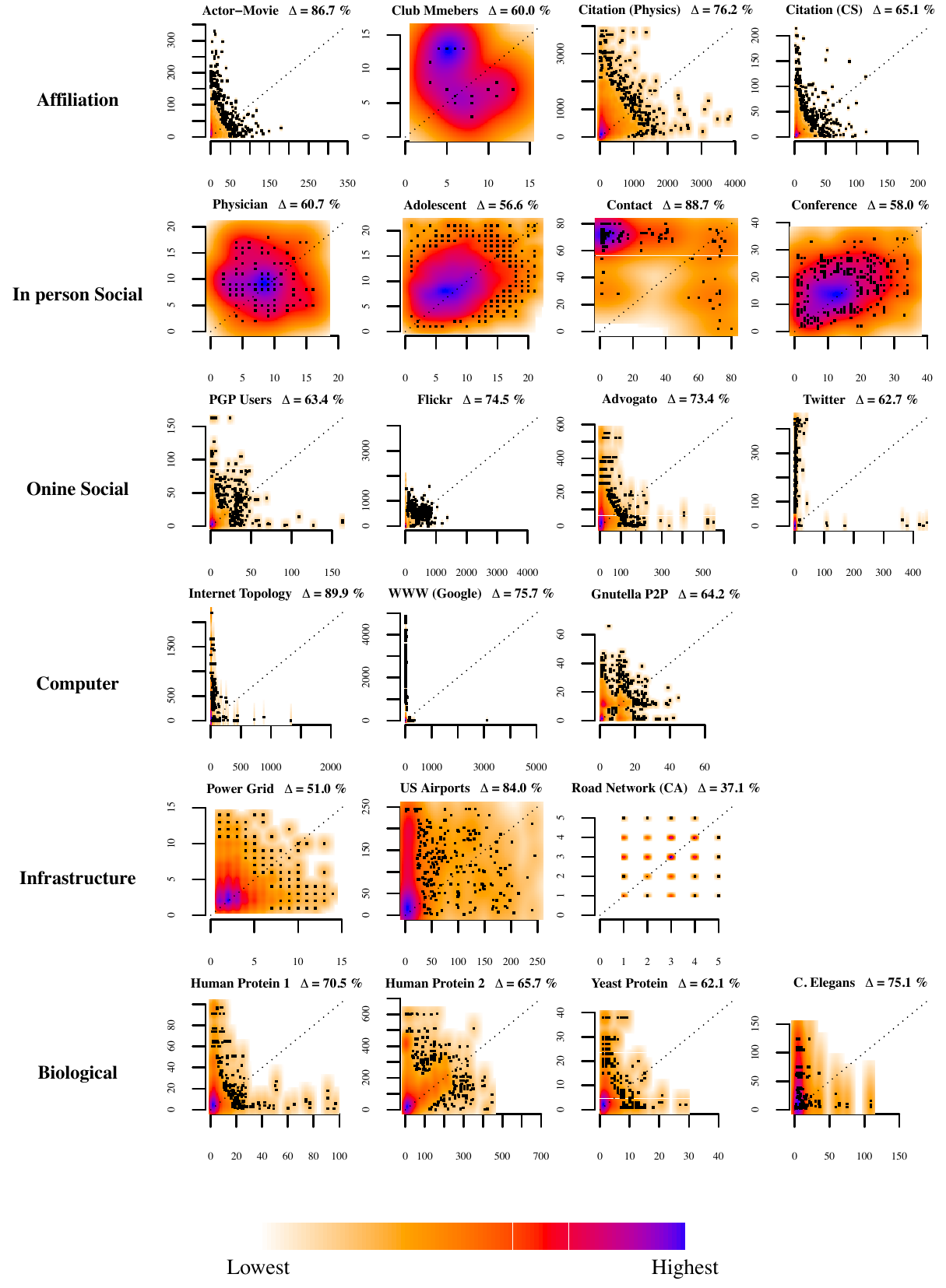}
\caption{\footnotesize{
Friendship Paradox at Individual Level}. Density plot of average number of friends of nodes compared to node degree in networks. $\Delta$ indictes the proportion of nodes that have a higher average number of friends of friends than their degree.  Lowest density regions within each network are marked by white / orange, and highest density regions are marked in blue. For all networks, the highest density region lies above and to the left of the 45 degree line. For some networks like Adolescent Health or Road Network (CA), it is relatively more evenly distributed both above and below the 45 degree line, whereas for networks like Internet Topology or Twitter, the distribution is skewed above and to the left.}
\label{figure:fpindividual}
\end{figure}
\renewcommand{\baselinestretch}{2} 

We illustrate this ``individual friendship paradox'' using a scatterplot of the node degree versus the average friend degree in Figure \ref{figure:individualfpcdf}.
Nodes that have a higher degree than their average friends do are colored red, whereas nodes that have lower degree are colored blue. Across most real networks, we observe that the blues vastly outnumber the reds. Relatedly, there are several nodes with low degrees whose friends on average have a high degree.

\renewcommand{\baselinestretch}{1} 
\begin{figure}[H]
\begin{center}
\includegraphics[scale=1.2]{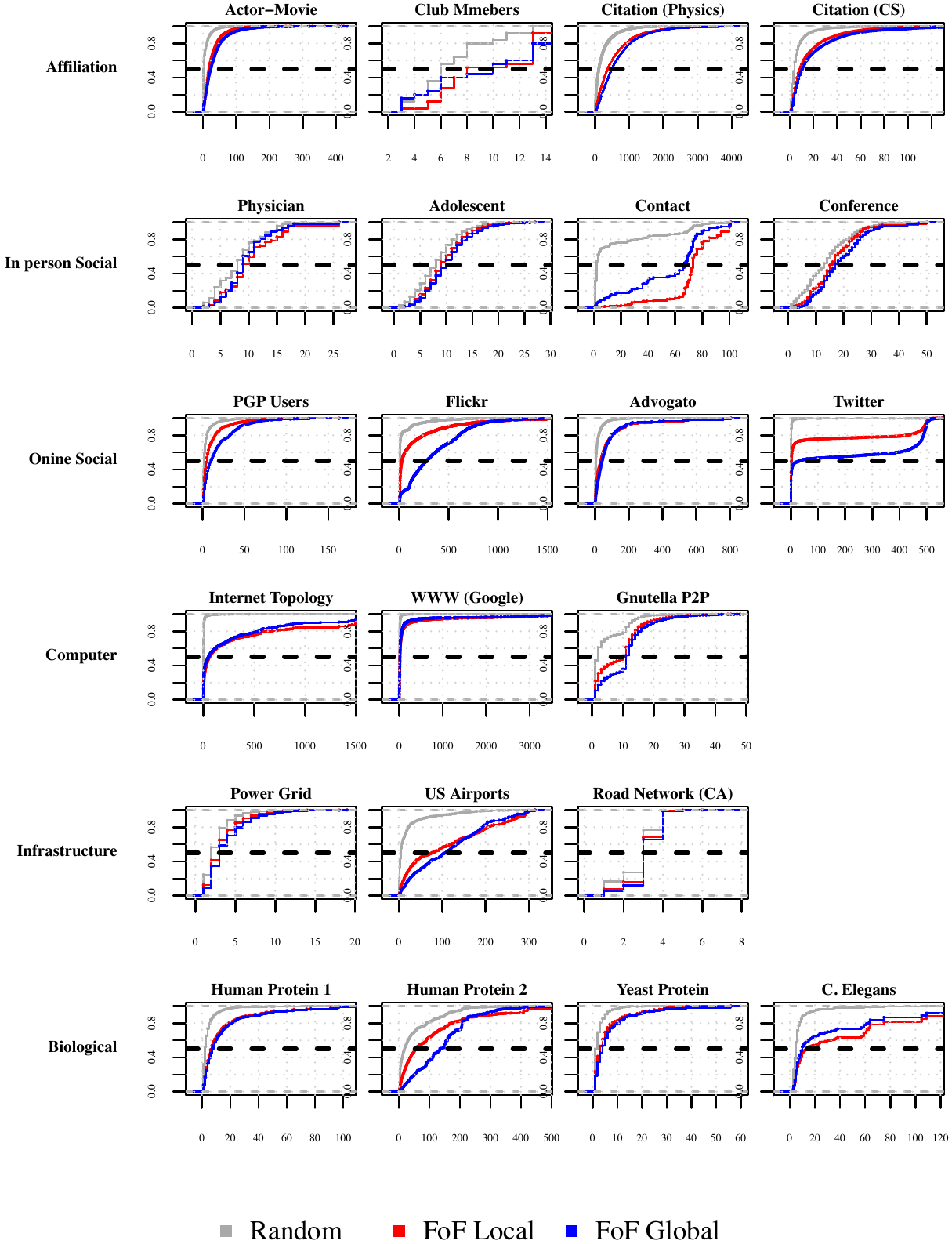}
\caption{Individual Friendship Paradox. Empirical Cumulative Distribution Functions (CDF) of Real Networks. Panels show the CDF of 3 different network properties at the individual node level. For a specific node degree, the probability that a node with a lower (or identical) degree is chosen by the sampling strategy for random sampling (gray), local FoF sampling (red) and global FoF sampling (green). Across all networks, for lower degrees, the random sampling curve is to the left of the local and global FoF curves. In several networks, global FoF is to the left and higher than local FoF (e.g. Contact), whereas in others, it is to the right (e.g. Flickr).}
\label{figure:individualfpcdf}
\end{center}
\end{figure}
\renewcommand{\baselinestretch}{2}

\section{Local and Global Means}\label{section:slocalglobalmeans}

We illustrate the practical impact of the distinction between the two means, with the following questions : (a) Is the Local Mean always greater (or smaller)  than the Global Mean? (b)  Can both means be relatively high (or low)? (c) What network (sub)structures result in a high Local Mean or Global Mean?

We examine four illustrative network structures (Fig. \ref{figure:two-by-two}) to answer these questions and to understand the differences between the two means. We find that both local and global mean can be much greater than the mean degree, and between these two means, either of them can be greater than the other.
Especially noteworthy is the difference between the Local and Global panels for network (Fig. \ref{figure:two-by-two}C): the Local mean is equivalent to the mean degree, because each node is equally weighted in terms of $w^L_i$. However, the Global mean is higher for this network since it assigns a higher weight $w^G_i$ to higher degree nodes.
\newpage

\renewcommand{\baselinestretch}{1}
\begin{figure}[H]
\centering{
\includegraphics[scale=0.84]{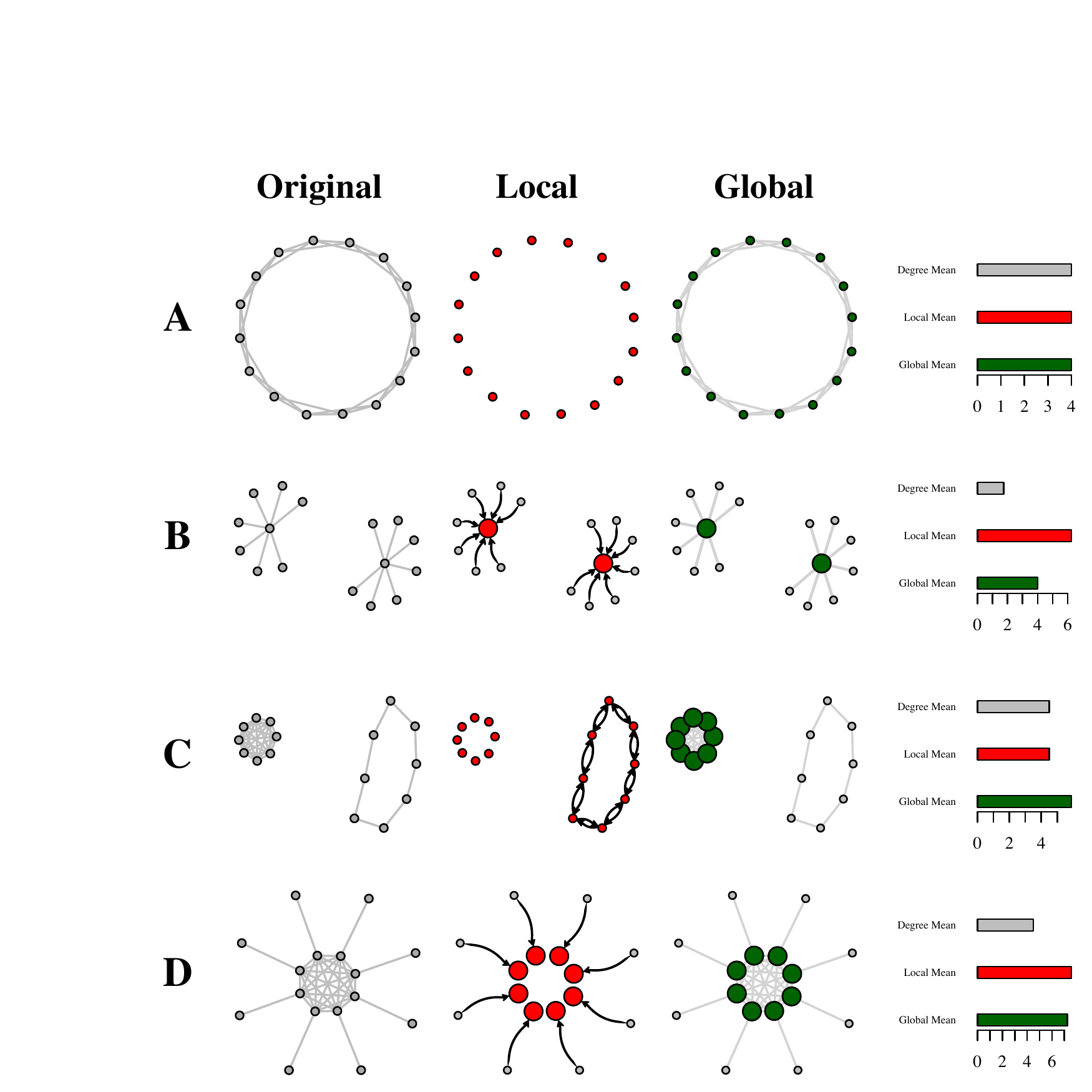}
}
\caption{\fcaptionsize{
\textbf{Four Illustrative Networks with Varying Local and Global Means.}\\
Each network in (A)-(D) has the original network plot (left), local weighted network (middle) and global weighted network (right). On the right is a barplot indicating the mean degree, local mean and global mean for each of the networks. 
\textit{Local Panel (Red)}: In the local weighted network plot (middle), nodes are sized proportional to their weight $\left(w^L_i\right)$ in contributing to the local mean. Edges that receive a \textit{higher than median} weight in computing the local mean are in black color. Otherwise, the edges are not plotted in the middle panel. Note that although the original networks are undirected, the selected edges are illustrated as \textit{directed} since the weights are directed.
\textit{Global Panel (Green)}: Nodes are sized proportional to their weight $\left(w^G_i\right)$ in contributing to the global mean. Edges are all weighted equally in the global weighted network. 
(A) \textit{Small World Ring:}  Each node has four friends, and local and global mean are both equal to avert age degree (4). None of the edges are shown in the middle panel since all edges have identical weight in computing the local mean.  All nodes in both local and global means have the same weight, and size in the middle and right panel.
(B) \textit{Two Central Hubs with Spokes:}  Each central hub is connected to 7 nodes. The mean degree is lowest in this network. However, local mean is substantially higher than the global mean, and is higher than the mean degree across all networks (a)-(d).  In local panel, we see that the weight of  central hubs has increased, whereas the corresponding weight for the low degree ``spoke'' nodes has decreased. In the global panel, the node weights are proportional to degree.
(C) \textit{Heavy Core with Attached Cycle:} The global mean is substantially higher than the local mean (and mean degree). Here, we see in the local panel that the weight of each of the nodes has not changed, and all nodes have the same weight. However, in the global panel, we see that the high degree nodes in the complete graph has higher weight compared to the original network, whereas the weights for the nodes in the 2-cycle are lower than in the original network.
(D) \textit{Heavy Core with Pendants:} Both the local and global mean are substantially higher than mean degree. In the local panel, the edges connecting core nodes to other nodes (both core and pendant) have a relatively low weight, and are not displayed.
}}
\label{figure:two-by-two}
\end{figure}
\renewcommand{\baselinestretch}{2}

\newpage
\section{Leverage and Network Size}\label{section:leveragenetworksize2}

\renewcommand{\baselinestretch}{1} 
\begin{figure}[H]
\begin{centering}
\includegraphics[width=0.5\textwidth]{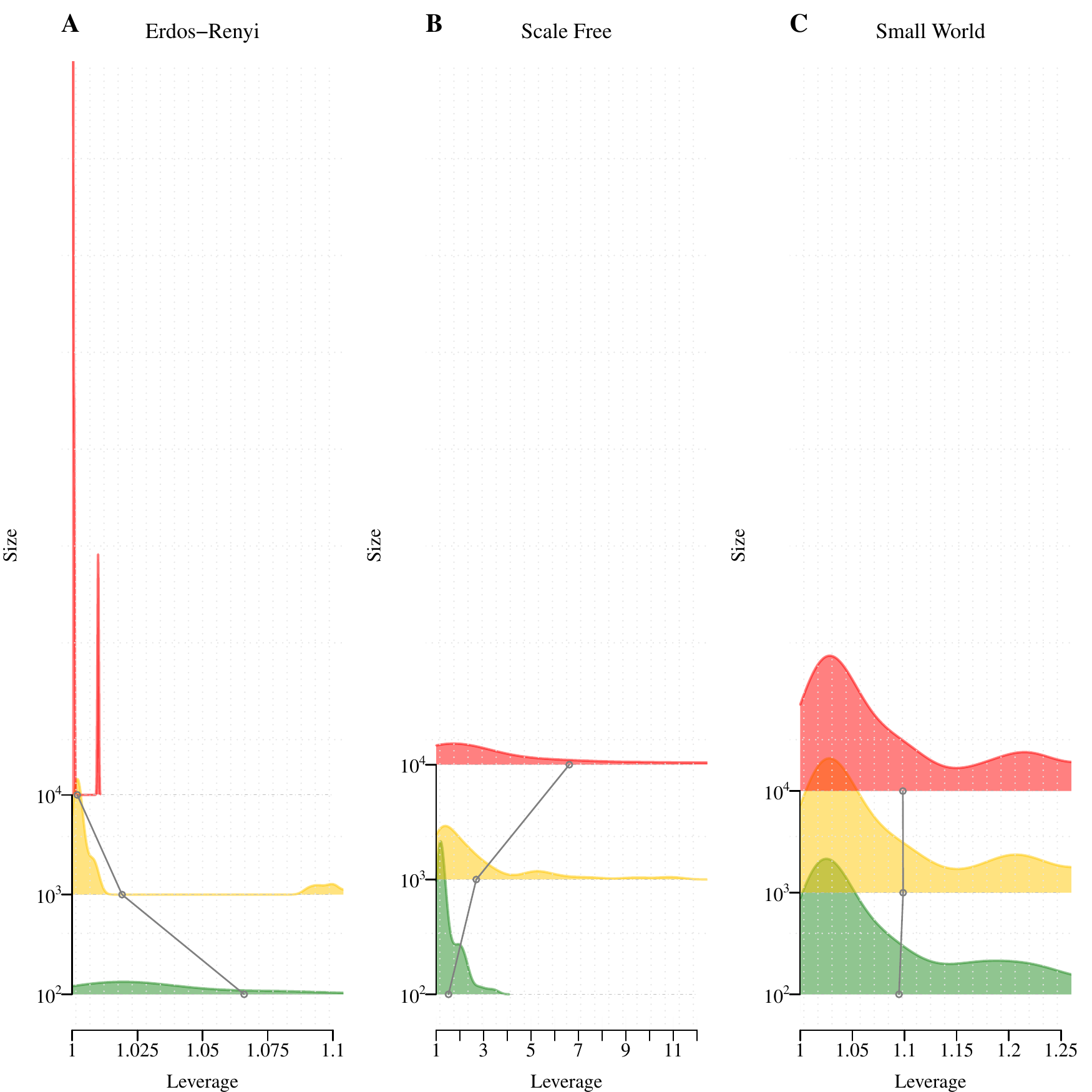}
\par\end{centering}
\caption{\fcaptionsize{
Local Leverage Density in Generated Networks from three different generative models, and spans the parameter space. A sample of 1,000 networks was used for each of the models with size varying between 100 and 10,000 nodes. (A) Erdos-Renyi (ER) networks generated with edge probabilities, $p \in [0.05,0.95]$. We find that as the size increases, the leverage decreases. In all cases, for ER networks, we find relatively low local leverage. (B) Static Scale Free (BA or Barabasi Albert) networks with scale-free parameter $\gamma \in [1,6]$. For these networks, we find that leverage increases as the size of the network increases, with networks of size 10,000 having an average leverage of over 7. (C) With small world (Watts-Strogatz) networks, the network size does not seem to materially impact leverage.
}}
\label{figure:leveragegeneratedsize}
\end{figure}
\renewcommand{\baselinestretch}{2} 

\newpage

\section{How Is Inversity Different from Degree Assortativity?}\label{section:inversityassortativity}
A natural question is whether inversity captures the same information (with opposite sign) as degree assortativity, which is a well known network property $\rho_a=Corr\left(D^{\mathbf{O}},D^{\mathbf{D}}\right)$ capturing the correlation in degree across all edges in the network \cite{newman2002assortative,newman2003social,sendina2016assortativity}. To examine this question, we generate $1,000$ networks using different generative methods as above. We find that assortativity and inversity are not guaranteed to have opposite signs (Fig. \ref{figure:inversityassortativity}A). Therefore, the sign of assortativity cannot be used to determine whether the local or global mean is greater for a network, unlike with inversity. All 3 network generating processes create networks with the same sign for both metrics (detail in Fig. \ref{figure:inversityassortativity}B). Example networks for the case of same sign assortativity and inversity are illustrated (Fig. \ref{figure:inversityassortativity}C, \ref{figure:inversityassortativity}D), showing that it is not obvious to predict inversity of a network if we know its assortativity.

\renewcommand{\baselinestretch}{1} 
\begin{sidewaysfigure}
\begin{centering}
\includegraphics[scale=0.75]{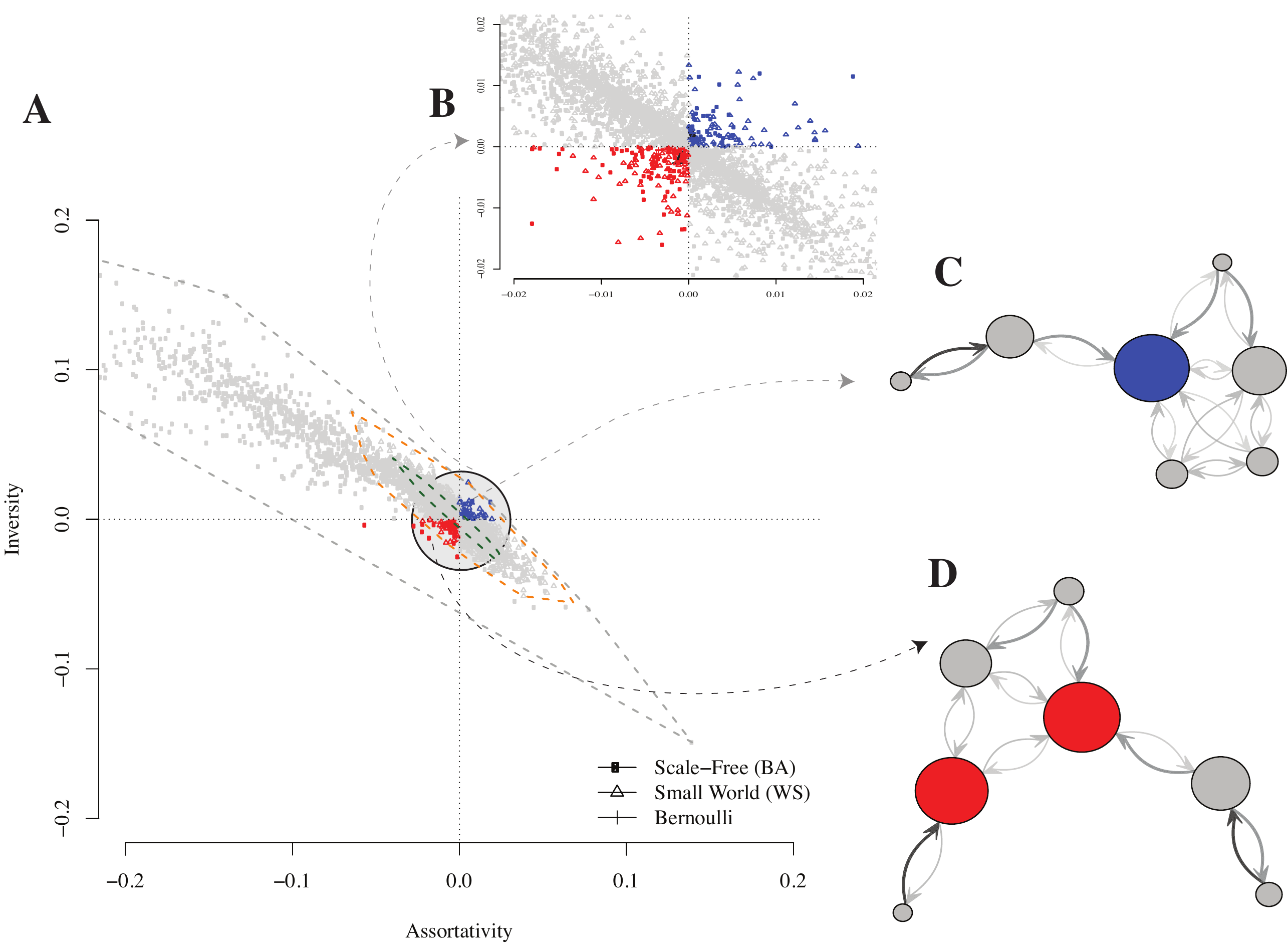}
\caption{
\fcaptionsize{
Assortativity and Inversity. N=1000 networks are generated from three classes of networks. (A) Erdos-Renyi (ER), Scale Free (BA or Barabasi Albert) and Small World (WS or Watts Strogatz), parameters detailed in legend of Figure 5. Observe the regions in red and blue, where networks have the same sign of assortativity and inversity. (B) Detailed view of region around (0,0) showing all three network types can produce networks with same sign for assortativity and inversity. (C) Example network with N=7 nodes where assortativity and inversity are both positive. (D) Similar example where both measures are negative. Overall, it demonstrates that inversity and assortativity are not equivalent measures, e.g. using assortativity in place of inversity could result in using a global strategy when local may be more appropriate.
}}
\label{figure:inversityassortativity}
\par
\end{centering}
\end{sidewaysfigure}
\renewcommand{\baselinestretch}{2} 

\newpage

\section{Virus Propagation Models}\label{ssection:vpm}

We detail below several examples of virus propagation models being used for characterizing the transmission and spread of diseases. These compartmental build upon the early work of Kermack and McKendrick \cite{m1925applications}. Thus, all individuals in a population (in our case, the nodes in a network) are in one of the states, either susceptible (S) or infected (I). Based on the viral propagation, they can move to other states like Exposed (E), Recovered (R), or Deceased (D). For example, the SIR model involves individuals being in one of three states, (S), (I) or (R) and transitioning between the states probabilistically. Typically, the vast majority of nodes are present in the susceptible state (S), in which they might contrast the disease. The exposed state (E) is used to indicate a node that has been exposed to the disease, but could be asymptomatic during an incubation period and is not capable of infecting others. In contrast, the infected state (I) indicates a node that is capable of infecting others. The (R) recovered state implies permanent immunity. 
There are further extensions possible, e.g. adding infants who have maternal antibodies (state M) that provide passive immunity. See \cite{brauer2005kermack} or \cite{hethcote2000mathematics} for an overview and survey of these models. These models have been extensively used in epidemiological studies to characterize disease dynamics as detailed in Table \ref{table:vpm}, including measles, influenza and COVID-19.

There has been recent notable work that aims to characterize the epidemic thresholds of these compartmental models with disease transmission over a network \cite{chakrabarti2008epidemic, prakash2010got}. The critical idea is that the epidemic threshold of a network can be characterized as the inverse of the greatest (first) eigenvalue of the adjacency matrix $A$ of the network, denoted as: $$\boxed{\tau(A) = \dfrac{1}{\lambda_1(A)}}$$

Eigenvalue $\lambda_1$ termed the spectral radius characterizes the connectivity of the network graph. Thus, networks that have higher connectivity or $\lambda_1$ are more likely to allow a contagion different paths to grow into an epidemic, whereas in networks with low connectivity, the epidemic is more likely to die out.

While there have been a number of epidemic thresholds for specific network generating processes (e.g. small world), the generality of the result above is valuable since it allows: (a) any arbitrary network, without placing restrictions on its topology or structure, (b) a wide range of compartmental models like SIS, SIR and others detailed in Table \ref{table:vpm} typically used to model infectious disease.

Consider a SIR model for illustration, the results also hold for the other models. The model is parametrized by two rates: $\beta$ is the probability of an infected node infecting a susceptible node in a given time period, and $\delta$ is the probability at which an infected node recovers (or is cured) during the period. If time is continuous, $\beta$ and $\delta$  can be viewed as the rates of infection and recovery. In either case, $ \mathcal{R}_0$ is defined as $ \boxed{ \mathcal{R}_0 = \dfrac{\beta}{\delta}}$.

The epidemic threshold $\tau$ is defined as follows \cite{chakrabarti2008epidemic}:\\
$$\boxed{
  \begin{cases}
    \mathcal{R}_0 = 
\dfrac{\beta}{\delta} < \tau(A) \implies \text{infection dies out over time}\\
    \mathcal{R}_0 =   
\dfrac{\beta}{\delta} > \tau(A) \implies \text{infection grows over time}\\
  \end{cases}
}
$$

There are a few observations relevant here. First, the critical value of epidemic threshold is a function of the adjacency matrix $A$ of the network topology (structure) $\mathcal{G}$. Second, a network topology with a higher epidemic threshold is less likely to have an epidemic. Third,  interventions like immunizing nodes or reducing the number of connections (edges) can increase the threshold $\tau(A)$ so that infections are more likely to die out.
\renewcommand{\baselinestretch}{1.5} 
\begin{table}[H]%
\caption{Virus Propagation Models Used for Diseases}\label{table:vpm}
\begin{tabular}{p{0.36\textwidth}p{0.54\textwidth}}
\toprule
\textbf{Virus Propagation Model}  & \textbf{Infectious Diseases [References]}  \tabularnewline
\midrule
SIS& Malaria (\cite{smith2005entomological})
\tabularnewline
SIR & Measles \cite{ferrari2008dynamics}, Swine Flu H1N1 \cite{small2005clustering}, Ebola \cite{berge2017simple}
\tabularnewline
SEIR& Chicken Pox \cite{deguen2000estimation}, SARS \cite{riley2003transmission}, COVID-19 \cite{prem2020effect}
\tabularnewline
SIRD& COVID-19 (\cite{caccavo2020chinese})
\tabularnewline
\bottomrule
\end{tabular}\\
Note: The states refer to (\textbf{S})usceptible, (\textbf{I})nfectious, (\textbf{R})ecovered / (\textbf{R})emoved, (\textbf{E})xposed, (\textbf{D})eceased 
\end{table}
\renewcommand{\baselinestretch}{2} 

\subsubsection*{Implementation of VPM} \label{ssection:vpmimplementation}
We begin with a seed set of 1\% of the nodes being infected, and evaluate epidemic outcomes using the SIR model. All the nodes in the network that are not infected or recovered are susceptible (S) to the infection.
Each infected node can transmit an infection in each period probabilistically to each of its neighbors. The probability of an infection is $P_\text{transmit} =\beta$. Thus, a node can become infected (I) from contact with any of its neighbors. In each period, an infected node can be cured or recovered (R) probabilistically, with the likelihood $P_\text{cure}  = \delta$. Recovered nodes cannot be reinfected and cannot transmit infections.

The process of immunizing (or vaccinating) a set of nodes involves choosing a proportion of nodes (5\%, or 10\% or 20\%) and ensuring that these nodes do not transmit any disease. The nodes for immunization are chosen based on three strategies: random, local and global.

The parameters used in the simulation of the epidemic are detailed in Table \ref{table:vpmsimulation}.

\renewcommand{\baselinestretch}{1.5} 
\begin{table}[H]%
\caption{Parameters of SIR Network Propagation Model }\label{table:vpmsimulation}
\begin{tabular}{p{0.15\textwidth}p{0.25\textwidth}p{0.60\textwidth}}
\toprule
\textbf{Parameter}  & \textbf{Value} & \textbf{Description}  \tabularnewline
\midrule
$P_\text{transmit} = \beta$& 0.20 & Probability of an infected node transmitting the disease to a susceptible neighbor.
\tabularnewline
$P_\text{cure} = \delta$& 0.15 & Probability of an infected node recovering. Thus, moving from (I) $\implies$ (R) is $P_{I\rightarrow R} = P_\text{cure}$, and  $P_{I\rightarrow I} = 1-P_\text{cure}$
\tabularnewline
$P^{k}_{S\rightarrow I}$& $1-(1- \beta)^{N^{\text{infected}}_k}$ & Probability of a susceptible node $k$ becoming infected. Depends on the number of infected neighbors $N^{\text{infected}}_k$. Thus, $k$ can become infected through \textit{any} of its infected neighbors. So we have: $P^{k}_{S\rightarrow I} = 1-(1-P_\text{transmit})^{N^{\text{infected}}_k}$. Similarly, $P^{k}_{S\rightarrow S} = (1-P_\text{transmit})^{N^{\text{infected}}_k}$.
\tabularnewline
$n^{0}_{\text{infected} }$& 1\% & Proportion of nodes in network that are infected at the beginning
\tabularnewline
$n_{\text{sim}}$& 100 & Number of simulations
\tabularnewline
\bottomrule
\end{tabular}
Note: (\textbf{S})usceptible, (\textbf{I})nfectious, (\textbf{R})ecovered / (\textbf{R})emoved
\end{table}
\renewcommand{\baselinestretch}{2}

Thus, a strategy $A$ is better than an alternative strategy $B$ if it results in lowers levels of peak infections, total infections and total suffering.

\section{Epidemic Outcomes on Facebook Network}\label{ssection:fbepidemic}
In Figure \ref{figure:epidemicfb}, we examine the epidemic propagation characteristics on the Facebook network \cite{viswanath09} using the same parameters as detailed in Table \ref{table:vpmsimulation}. For the Facebook network, we find that an epidemic's outcomes are better when using the local strategy compared to the global strategy, which in turn in are better than the random strategy.

\renewcommand{\baselinestretch}{1}
\begin{figure}
\begin{centering}
\includegraphics[scale=0.75]{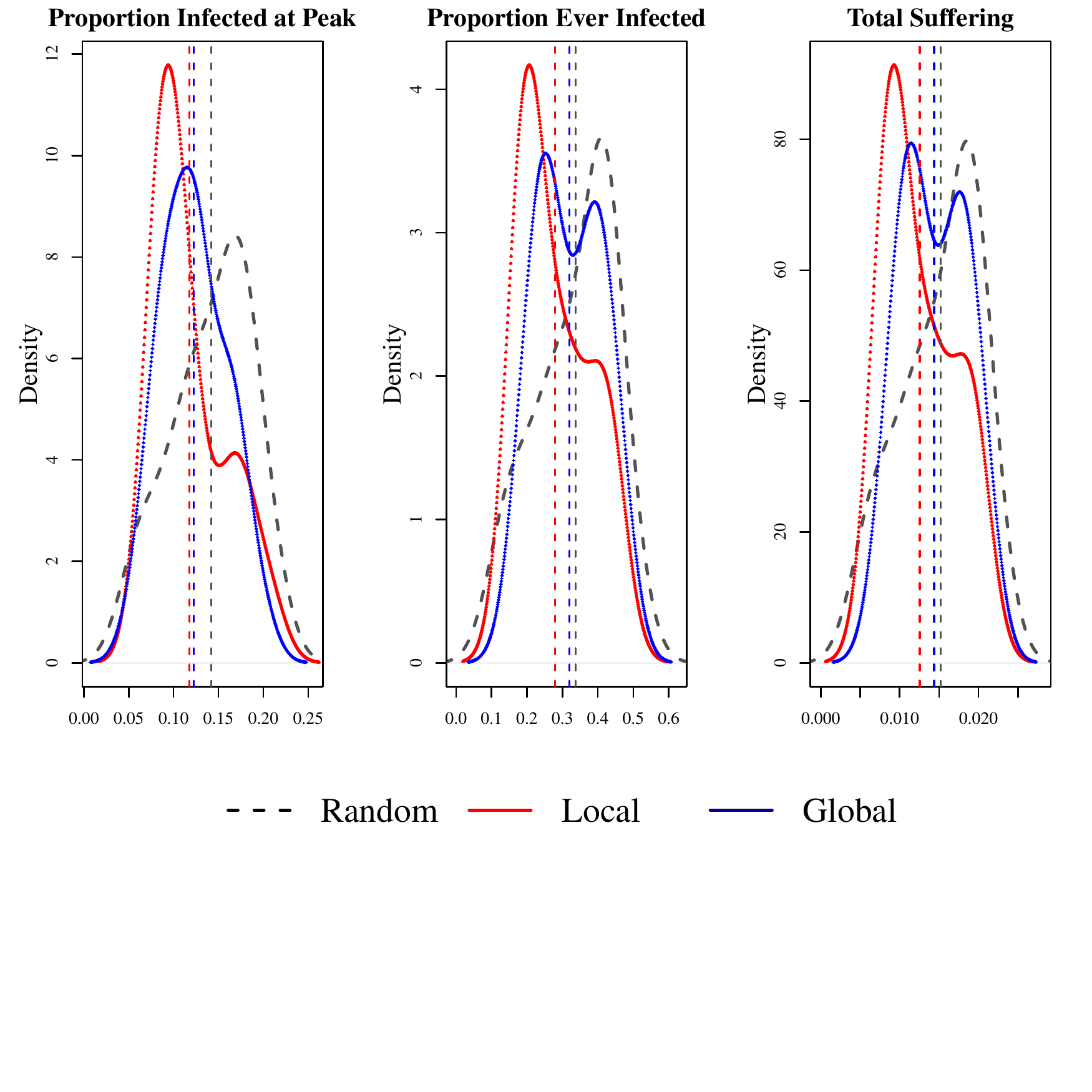}
\end{centering}
\begin{centering}
\caption{Epidemic Characteristics with Immunization in Facebook Network. See Table {table:vpmsimulation} for parameters of simulation. All outcomes are density plots. We plot 3 outcomes: (a) the proportion of population infected at the peak, (b) proportion of population that was ever infected, and (c) total suffering. The $x$-axis represent proportions and the $y$-axis represent density.  We plot the outcomes for 3 strategies: (R)andom, (L)ocal and (G)lobal. The dashed vertical lines represent the means for the 3 strategies. We find that for the Facebook network, the Local strategy is better for all outcomes than the Global, which in turn is better than the Random strategy.}
\label{figure:epidemicfb}
\end{centering}
\end{figure}
\renewcommand{\baselinestretch}{2}

\newpage

\end{document}